\documentclass[10pt,letterpaper]{article}
\usepackage[utf8]{inputenc}
\usepackage[bookmarks]{hyperref}
\hypersetup{colorlinks=true,citecolor=blue,linkcolor=blue,filecolor=blue,urlcolor=blue}

\usepackage{fullpage}
\usepackage{tikz}
\usepackage{circuitikz}

\usepackage[affil-it]{authblk}

\usepackage{amsmath,amsfonts,amssymb,amsthm,mathtools}

\usepackage{cleveref}

\newtheorem{thm}{Theorem}\crefname{thm}{Theorem}{Theorems}
\newtheorem{prop}[thm]{Proposition}\crefname{prop}{Proposition}{Propositions}
\crefname{lem}{Lemma}{Lemmas}
\newtheorem{cor}[thm]{Corollary}\crefname{cor}{Corollary}{Corollaries}
\newtheorem*{thm-main}{Theorem}\crefname{thm}{Theorem}{Theorem}
\theoremstyle{definition}
\newtheorem{rem}[thm]{Remark}\crefname{ref}{Remark}{Remarks}

\usepackage[backend=bibtex8,sorting=none,firstinits=true,doi=false,isbn=false,url=false,maxbibnames=6,maxcitenames=2,style=nature]{biblatex}
\renewbibmacro{in:}{}
\usepackage{graphicx}
\usepackage[labelformat=simple]{subcaption}

\usepackage{tikz}
\usepackage{pgfplots}

\usepackage{xifthen}

\newcommand{\cA}{\mathcal{A}}
\newcommand{\cB}{\mathcal{B}}
\newcommand{\cC}{\mathcal{C}}

\newcommand{\cM}{\mathcal{M}}
\newcommand{\cN}{\mathcal{N}}

\newcommand{\cR}{\mathcal{R}}
\newcommand{\cS}{\mathcal{S}}
\newcommand{\cT}{\mathcal{T}}

\newcommand{\cX}{\mathcal{X}}
\newcommand{\cY}{\mathcal{Y}}
\newcommand{\cZ}{\mathcal{Z}}

\newcommand{\ox}{\otimes}
\newcommand{\eps}{\varepsilon}
\DeclareMathOperator{\supp}{supp}
\DeclareMathOperator{\poly}{poly}
\DeclareMathOperator{\id}{id}
\DeclareMathOperator{\prob}{Pr}
\newcommand{\dtv}{d_{\mathrm{TV}}}
\DeclareMathOperator{\tr}{tr}
\newcommand{\CeaO}[1][]{
	\ifthenelse{\isempty{#1}}
	{\cC^{(1)}_{\mathrm{ea}} }
	{\cC^{(1)}_{\mathrm{ea},#1}}
}
\newcommand{\Cea}[1][]{\cC_{\text{ea},#1}}

\newcommand{\Csum}{S}
\newcommand{\CsumEA}[1][]{S_{\mathrm{ea},{#1} } }

\newcommand{\MS}{\mathrm{MS}}
\newcommand{\SV}{\mathrm{SV}}
\newcommand{\LS}{\mathrm{LS}}
\newcommand{\Ha}{\mathrm{H}}
\newcommand{\Nms}{N_{G_{\MS}}}
\newcommand{\ii}{\mathrm{i}}

\title{Playing Games with Multiple Access Channels}
\author[a,b]{Felix Leditzky\thanks{Corresponding author; email: \texttt{felix.leditzky@jila.colorado.edu}}}
\author[a,c]{Mohammad A. Alhejji}
\author[a,c]{Joshua Levin}
\author[a,b,c]{Graeme Smith}
\affil[a]{\small JILA, University of Colorado/NIST, Boulder, CO 80309, USA}
\affil[b]{\small Center for Theory of Quantum Matter, University of Colorado, Boulder, CO 80309, USA}
\affil[c]{\small Department of Physics, University of Colorado, Boulder, CO 80309, USA}

\begin{document}
\maketitle

\begin{abstract}
	Communication networks have multiple users, each sending and receiving messages.  A multiple access channel (MAC) models multiple senders transmitting to a single receiver, such as the uplink from many mobile phones to a single base station. The optimal performance of a MAC is quantified by a capacity region of simultaneously achievable communication rates. We study the two-sender classical MAC, the simplest and best-understood network, and find a surprising richness in both a classical and quantum context.  First, we find that quantum entanglement shared between senders can substantially boost the capacity of a classical MAC.  Second, we find that optimal performance of a MAC with bounded-size inputs may require unbounded amounts of entanglement.  Third, determining whether a perfect communication rate is achievable using finite-dimensional entanglement is undecidable. Finally, we show that evaluating the capacity region of a two-sender classical MAC is in fact NP-hard.
\end{abstract}

\section{Introduction}

Information theory is the mathematical theory of communication and signal processing pioneered by Shannon \cite{Sha48}.
In network communication settings, the simplest model is a multiple access channel (MAC), where two spatially separated senders aim to transmit individual messages to a single receiver.
Faithful information transmission through a MAC is possible within its capacity region, which was characterized by Ahlswede \cite{Ahlswede71} and Liao \cite{Liao72} in terms of a so-called single-letter formula, i.e., an entropic optimization problem of fixed bounded dimension that is in principle computable.
In quantum information theory, communication tasks can be enhanced dramatically if the communicating parties are given access to quantum resources such as shared entanglement \cite{BW92,BBC+93}.
However, certain tasks such as classical single-sender-single-receiver communication receive no advantage from entanglement assistance \cite{BSST99}.

In this work, we show that MACs behave in a fundamentally different way in the presence of entanglement assistance, in contrast to the single-sender-single-receiver scenario.
Moreover, even unassisted classical MACs exhibit far more complex behavior than previously widely appreciated. 
We demonstrate this by constructing a family of classical MACs with surprisingly rich behavior:
First, we show that entanglement shared between the senders can strictly increase the capacity region of a classical MAC, proving that entanglement can help in a purely classical communication scenario.
Second, we exhibit examples of channels for which an unbounded amount of entanglement is needed to achieve the maximal possible increase of the achievable rate region.
We also show that it is generally undecidable to determine whether the maximal rate pair can be achieved for a MAC with finite-dimensional entanglement strategies.
Finally, we prove in the unassisted communication setting that it is NP-hard to determine which rates can be achieved for a given MAC.
Our findings imply that even for the arguably simplest network information-theoretic setting of a MAC, there is no general solution to the problem of determining its unassisted communication capabilities, highlighting the need for practical approximation algorithms.
At the same time, entanglement assistance can push the achievable information transmission rates of MACs beyond the classical limit, paving the way for harnessing entangled resources in hybrid classical-quantum information networks. 

\section{Results}

\subsection{Entanglement Helps a Classical Multiple Access Channel}\label{sec:entanglement}
We first briefly review classical multiple access channels (MAC), a general example of which is shown in \Cref{fig:macs}.
Our results concern the capacity region $\cC(N)$ of a MAC $N$, consisting of all the rate pairs $(R_1,R_2)$ such that sender $i$ can faithfully transmit information to the receiver at the rate $R_i$ (see \Cref{sec:macs} for a more detailed definition).
Ahlswede \cite{Ahlswede71} and Liao \cite{Liao72} proved that $\cC(N)$ is given by the convex hull of all pairs $(R_1,R_2)$ satisfying
\begin{align}
	R_1 &\leq I(A; Z|B) & R_2 &\leq I(B; Z|A) & R_1 + R_2 &\leq I(AB;Z)
	\label{eq:MAC-capacity-region-main}
	\end{align}
for some product distribution $\pi_A\pi_B$ on $\cA\times \cB$.
Here, $I(U;V|W) = H(UW) + H(VW) - H(W) - H(UVW)$ is the conditional mutual information,  $H(X) = -\sum_i p(x_i) \log p(x_i)$ is the Shannon entropy of a random variable $X\sim p$ with the logarithm taken to base $2$, and $I(U;V) = H(U)+H(V)-H(UV)$ is the mutual information.

The central object in our work is a multiple access channel $N_G$ defined in terms of a non-local game $G=(\cX_1,\cX_2,\cY_1,\cY_2,W)$, where $\cX_1,\cX_2$ and $\cY_1,\cY_2$ are the question and answer sets for Alice and Bob, respectively, and $W\subset \cX_1\times\cX_2\times\cY_1\times\cY_2$ is the winning condition.
For the MAC $N_G$, the input alphabets of the two senders Alice and Bob are the question-answer sets $\cX_1\times\cY_1$ and $\cX_2\times\cY_2$, respectively, and the output alphabet of $N_G$ is $\cX_1\times\cX_2$.
If Alice and Bob win the non-local game $G$, that is, $(x_1,y_1,x_2,y_2)\in W$, the channel is noiseless and outputs the question pair $(x_1,x_2)$ to the receiver.
If they lose the game, $(x_1,y_1,x_2,y_2)\notin W$, the channel outputs a question pair $(\hat{x}_1,\hat{x}_2)$ drawn uniformly at random from $\cX_1\times\cX_2$.
More formally, the MAC $N_G\colon (\cX_1\times\cY_1) \times (\cX_2\times\cY_2)\to \cX_1\times \cX_2$ is defined as
\begin{align}
N_G(\hat{x}_1,\hat{x}_2 | x_1,y_1; x_2,y_2) &\coloneqq \begin{cases}
\delta_{x_1\hat{x}_1}\delta_{x_2 \hat{x}_2} & \text{if }(x_1,x_2,y_1,y_2) \in W,\\
(|\cX_1| |\cX_2|)^{-1} & \text{else}.
\end{cases}
\label{eq:non-local-mac-main}
\end{align}
This channel construction is inspired by previous work by Quek and Shor\cite{QuekShor}, who used a similar construction in terms of the CHSH game\cite{CHSH69} for an interference channel consisting of two senders and two receivers.
It also appeared in unpublished work by N\"{o}tzel and Winter \cite{NW17}, a portion of which has appeared in \cite{Noe19}.

The noise in the MAC $N_G$ defined in \eqref{eq:non-local-mac-main} is determined by the players' ability to win the non-local game $G$.
Clearly, if there exists a perfect strategy for Alice and Bob (i.e., a strategy that wins the game with certainty on any question pair), they can select their questions uniformly at random and transmit information to the receiver at rates $R_i = \log |\cX_i|$, achieving the maximal possible sum rate $R_1 + R_2 = \log |\cX_1| + \log |\cX_2|$.
On the other hand, if they cannot win the game with certainty, then the channel necessarily adds noise to their signals, and consequently the achievable sum rate decreases.
We can make this intuition precise by observing that, setting $A=X_1Y_1$ and $B=X_2Y_2$, the mutual information $I(X_1Y_1X_2Y_2;Z)$ constraining the sum rate $R_1 + R_2$ in \eqref{eq:MAC-capacity-region-main} can be expressed as 
\begin{align}
I(X_1Y_1X_2Y_2;Z) = H(Z) - p_L (\log|\cX_1|+\log|\cX_2|).
\end{align}
Here, $p_L$ denotes the probability of losing the game $G$ given a product distribution $p_Ap_B$ on the questions $x_i$ and a strategy producing the answers $y_i$.
This relation allows us to prove a bound on the capacity region of $N_G$ whenever the non-local game $G$ does not admit a perfect strategy:
\begin{thm}
	If a non-local game $G$ does not admit a perfect strategy (using the available resources), the sum rate $R_1+R_2$ of the MAC $N_G$ defined in \eqref{eq:non-local-mac-main} is strictly bounded away from $\log |\cX_1| + \log|\cX_2|$.
	\label{thm:separation}
\end{thm}

Consider now a non-local game $G$ which cannot be won with certainty using any classical strategy, and assume there exists a perfect quantum strategy.
If we allow shared entanglement between the two senders of the MAC, then \Cref{thm:separation} provides a provable separation between the capacity region of the unassisted MAC $N_G$ and the entanglement-assisted achievable rate region.
A well-known example of such a non-local game is the magic square game $G_{\MS}$ \cite{Mer90,Peres1990,Ara02b,BBT05}, in which Alice and Bob have to fill a given row and column of a $3\times 3$-square with a bit string such that the parity of Alice's row is even, the parity of Bob's column is odd, and the assignments are consistent in the overlapping cell.
It is well-known that any classical strategy has winning probability at most $\frac{8}{9}$ \cite{BBT05}, as illustrated in the left panel in \Cref{fig:magic-square-game}.
For the MAC $\Nms$ defined in terms of the magic square game, we can use \Cref{thm:separation} to obtain an upper bound on the achievable sum rate of $3.13694$, bounding it away from the maximal value of $2\log 3 \approx 3.17$.

On the other hand, there is a perfect quantum strategy in which Alice and Bob make measurements on two maximally entangled states shared between them \cite{Mer90,Peres1990,BBT05}.
This perfect quantum strategy, depicted in the right panel of \Cref{fig:magic-square-game}, allows Alice and Bob to send individual messages to the receiver at the rates $R_1 = R_2 = \log 3$, yielding the maximal sum rate $R_1+R_2=2\log 3$ which is not achievable by any classical coding strategy by \Cref{thm:separation}.
Therefore, the unassisted capacity region of $\Nms$ is separated from the point $(\log 3,\log 3)$, while the entanglement-assisted achievable rate region includes this point.
This separation between classical strategies and entanglement-assisted strategies occurs for any non-local game $G$ with no perfect classical strategies and perfect quantum strategies, a so-called ``pseudo-telepathy game'' \cite{BBT05}.
Each game in this class yields a separation for the corresponding MAC between the unassisted capacity region and the entanglement-assisted region via \Cref{thm:separation}.

\subsection{How Much Entanglement Do You Need?}\label{sec:unbounded}

Our main result, \Cref{thm:separation}, can also be applied to separate achievable rate regions for coding strategies using different amounts of entanglement.
To illustrate this, we consider the class of linear system games $G_{\LS}$ \cite{cleve2014characterization}, which are defined in terms of an $m\times n$ linear system $Ax=b$ of equations over $\mathbb{F}_2$.
Slofstra and Vidick showed that there is a particular instance $G_{\SV}$ of a linear system game for which a perfect winning strategy is necessarily quantum and furthermore requires an unbounded amount of entanglement \cite{SV18}.
More precisely, they gave upper and lower bounds on the local dimension $d$ of the quantum systems associated with Alice and Bob in the quantum strategy in terms of the losing probability $p_L$.
Consider now the MAC $N_{G_{\SV}}$ defined according to \eqref{eq:non-local-mac-main} in terms of the linear system game $G_{\SV}$.
Limiting Alice and Bob to entanglement assistance of local dimension at most $d$, their probability of losing the linear system game is strictly positive \cite{SV18}.
Consequently, we can invoke \Cref{thm:separation} to conclude that the $d$-dimensional entanglement-assisted achievable rate region of $N_{G_{\SV}}$ is bounded away from the rate pair $(\log m, \log n)$ achieving the maximal sum rate $\log m + \log n$.
On the other hand, it is straightforward to define a $d$-dimensional entanglement-assisted coding strategy for Alice and Bob based on the quantum strategy derived by Slofstra and Vidick \cite{SV18} whose winning probability converges to unity as $d$ grows.
Hence, as Alice and Bob have access to larger and larger entangled states, they approximate the perfect sum rate $\log m + \log n$ arbitrarily well.

Our results show that linear system games give rise to a family of MACs whose $d$-entanglement-assisted achievable rate regions approach the rate pair $(\log m,\log n)$ in the limit $d\to\infty$, yet they are strictly bounded away from it for any fixed finite $d$.
Moreover, considering all finite-dimensional quantum strategies for a general linear system game $G_{\LS}$, Slofstra showed that it is undecidable to determine whether there is a perfect quantum strategy among them \cite{Slofstra2019}.
By the arguments above, this directly translates to the following result: For the MAC $N_{G_{\LS}}$ defined in terms of a linear system game $G_{\LS}$, it is undecidable to determine whether the entanglement-assisted achievable rate region includes the rate pair $(\log m, \log n)$.

\subsection{Complexity of the Capacity Region of a Classical MAC}\label{sec:np-hard}

Finally, we turn our focus to the unassisted coding scenario for a discrete MAC.
In information-theoretic terms, this scenario seems well understood as the capacity region $\cC(N)$ of a MAC $N$ can be expressed in terms of a computable single-letter formula \cite{Ahlswede71,Liao72}.
However, the single-letter nature of the capacity region formula by itself does not guarantee an efficient method of computing $\cC$ in, say, runtime polynomial in $\max\lbrace |\cA|,|\cB|,|\cZ| \rbrace$, the maximal size of the input and output alphabets of $N$.
Using our construction of a MAC in terms of non-local games, we prove that it is NP-hard to decide whether a given point $(R_1,R_2)$ belongs to the capacity region of a MAC to within an additive error inverse-cubic in $n$, where $n$ is the size of the output alphabet.
Our result is based on a non-local game version $G_{\Ha}$ of 3SAT introduced by Håstad \cite{H01}, consisting of $m=\mathcal{O}(n)$ clauses containing exactly three out of $n$ literals.
For this non-local game, it follows from the probabilistically checkable proofs (PCP) theorem \cite{F98, SS06} that it is NP-hard to decide if there is a perfect winning strategy for $G_{\Ha}$ or if the maximal winning probability is bounded from above by $1-(1-c)/n$ for some constant $c<1$ \cite{H01}.
Consider now $N_{G_{\Ha}}$, the MAC defined in \eqref{eq:non-local-mac-main} in terms of the game $G_{\Ha}$. 
Clearly, if Alice and Bob have a perfect winning strategy for $G_{\Ha}$, they can each code at the rates $R_1 = \log m$ and $R_2 = \log n$ by choosing a uniform distribution over their respective question alphabets.
This leads to the maximal sum rate $R_1+R_2=\log m + \log n$.
On the other hand, if the maximal winning probability is bounded from above by $\omega^*=1-(1-c)/n$, then \Cref{thm:separation} can be used to show that the sum rate $R_1+R_2$ is bounded from above by $\log m + \log n - (1-\omega^*)^3$.
In this case, the capacity region $\cC(N_{G_{\Ha}})$ is bounded away from the rate pair $(\log m,\log n)$.
Altogether, this shows that it is NP-hard to decide if an arbitrary rate tuple $(R_{1},R_{2})$ belongs to $\cC(N_{G_{\Ha}})$ to precision inverse-cubic in $n$.

\section{Discussion}

In this work we show that the capacity region of a multiple access channel displays complex behavior, both in a purely classical setting and when the senders have access to shared entangled quantum states.
In particular, we prove that entanglement assistance can boost the achievable rates in a setting where two senders try to convey classical information through a common classical communication channel to a single receiver.
Such an increase in capacity is impossible in the point-to-point scenario involving a single sender and a receiver.
We also show that for a certain family of MACs the two senders need to share an unbounded amount of entanglement in order to achieve the ideal communication rate pair.
When restricted to finite-dimensional entangled strategies, it is undecidable for this particular channel family whether the ideal rate can be achieved.
Finally, we show that even in the unassisted scenario, it is in fact NP-hard to decide whether the ideal rate pair belongs to the capacity region of a MAC.
This result is a strong counterpoint to the widely held belief that the availability of a computable single-letter formula for the capacity region essentially solves the MAC problem.
The central tool in the proofs of all results above is the construction of a MAC in terms of a non-local game in such a way that the noise level of the channel is determined by the senders' ability to win the game.

Our work opens up a number of interesting topics for future work.
Numerical investigations for the magic square channel suggest that the true separation between classical and quantum coding strategies for the MACs considered in this work is considerably larger than the separation guaranteed by \Cref{thm:separation}, suggesting that our bound on the sum rate could be further tightened.
For our results above we considered a specific achievable rate region that arises naturally when the two senders measure identical copies of a single entangled state.
In general, the senders might have access to multipartite entangled states and implement parallel encoding strategies, which leads to the notion of an entanglement-assisted capacity region.
We expect this region to be given by the regularization of the achievable region considered in this paper.
For the MACs defined via our construction, this question seems to be related to parallel repetition theorems for non-local games played with quantum strategies.
Furthermore, in this work we only considered entanglement shared between the two senders, and the communication setting could be generalized to one where entanglement is shared between both the senders and the receiver.
Finally, our NP-hardness result for the unassisted capacity region of a MAC underlines the need for tight efficiently computable outer bounds on the unassisted capacity region. 
Such bounds could for example be obtained from convex relaxations of the rank-1 optimization problem describing the MAC capacity region \cite{CPFV10}.

\paragraph*{Data availability}
No data sets were generated during this study.

\paragraph*{Code availability}
MATLAB and Mathematica code files used to obtain the numerical bounds in \Cref{sec:magic-square-game} are available from the corresponding author upon request.

\printbibliography[heading=bibintoc]

\paragraph*{Acknowledgments}
We thank Mark M.~Wilde and Andreas Winter for helpful comments and feedback.
This work was partially supported by NSF Grant No.~PHY 1734006 and NSF CAREER award CCF 1652560.

\paragraph*{Author Contributions}
FL, MA, JL, and GS each contributed extensively to the paper.

\paragraph*{Author information}
The authors declare no competing interests.
Correspondence should be addressed to \texttt{felix.leditzky@jila.colorado.edu}.

\begin{figure}[h]
	\centering
	\begin{subfigure}[b]{0.3\textwidth}
		\includegraphics[width=\textwidth]{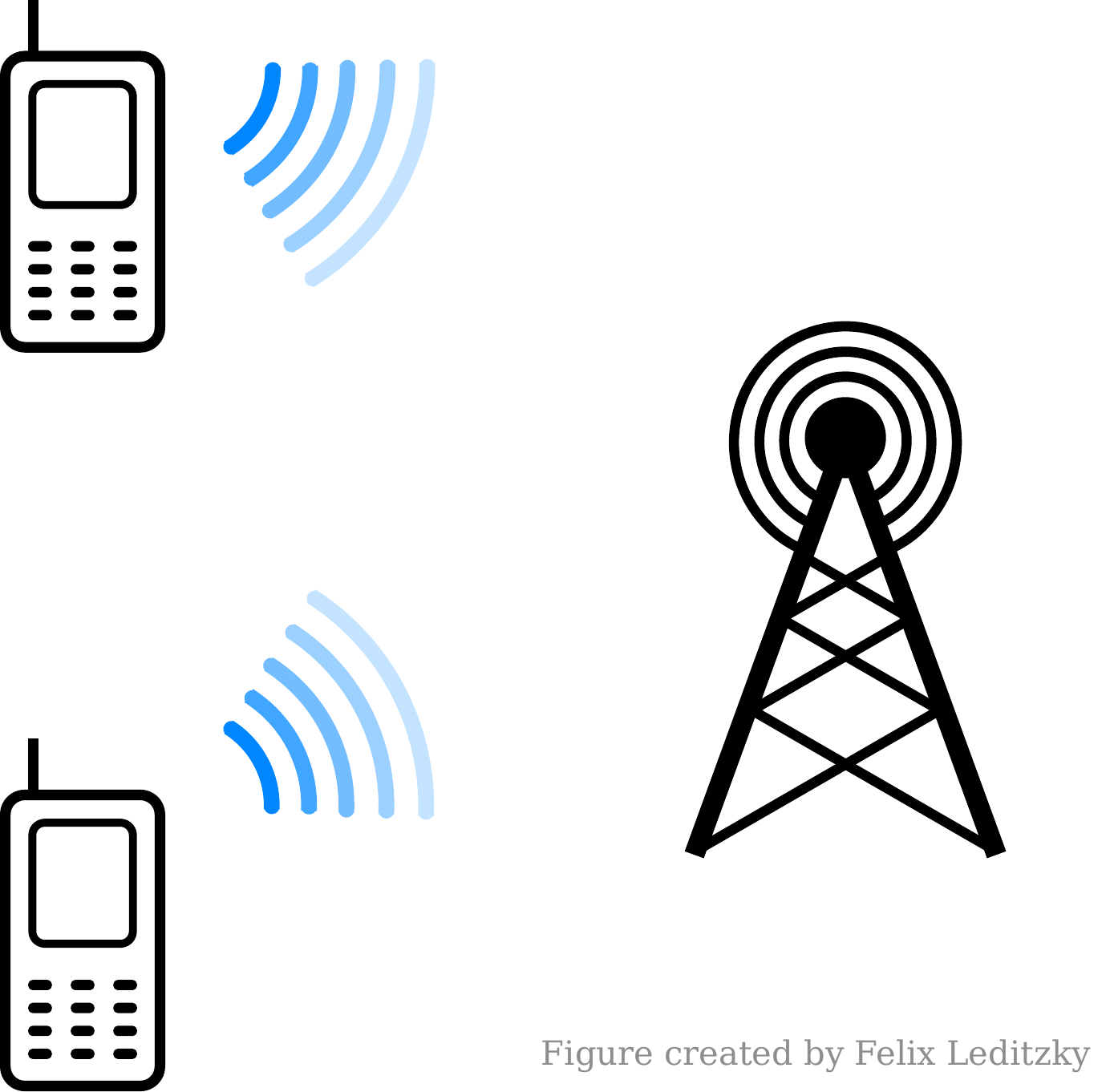}
		\caption{Realistic MAC model.}
		\label{fig:MAC-realistic}
	\end{subfigure}
	\hfill
	\begin{subfigure}[b]{0.3\textwidth}
		\centering
		\begin{tikzpicture}
		\draw[very thick] (0,2) rectangle (2,-1) node[pos=.5] {$N(z|a,b)$};
		\draw[thick] (2,.5) -- (3,.5) node[pos=.5, above right] {$Z$};
		\draw[thick] (-1,1.5) -- (0,1.5) node[pos=0, above right] {$A$};
		\draw[thick] (-1,-.5) -- (0,-.5) node[pos=0, below right] {$B$};
		\end{tikzpicture}
		\caption{Mathematical MAC model.}
		\label{fig:MAC-math}
	\end{subfigure}
	\hfill
	\begin{subfigure}[b]{0.3\textwidth}
		\centering
		\begin{tikzpicture}
		\draw[->,>=latex,very thick] (0,0) -- +(0,3) node[pos=1,above] {$R_1$};
		\draw[->,>=latex,very thick] (0,0) -- +(3,0) node[pos=1,right] {$R_2$};
		\draw[dashed,thick] (0,2) -- (1.2,2) -- (2,1.2) -- (2,0);
		\draw[thick] plot [smooth] coordinates {(0,2.25) (1.5,2) (2,1.5) (2.25,0)};
		\end{tikzpicture}
		\caption{Capacity region of a MAC.}
		\label{fig:cap-region}
	\end{subfigure}
	\caption{
		Multiple access channels.\\
		\subref{fig:MAC-realistic} Realistic scenario of a multiple access channel (MAC), in which two cell phones send data to a cell tower. 
		\subref{fig:MAC-math} Mathematical model of a MAC $N$, characterized by finite input alphabets $\cA$ and $\cB$, an output alphabet $\cZ$, and a conditional probability distribution $N(z|a,b)$ for $a\in \cA, b\in\cB, z\in\cZ$. 
		The random variables corresponding to the senders and the receiver are denoted by $A$, $B$, and $Z$, respectively.
		\subref{fig:cap-region} A typical capacity region of a MAC (solid line), together with an achievable pentagonal region for a fixed input distribution (dashed lines).
	}
	\label{fig:macs}
\end{figure}

\begin{figure}[h]
	\centering
	\begin{subfigure}[b]{0.4\textwidth}
		\centering
		\includegraphics{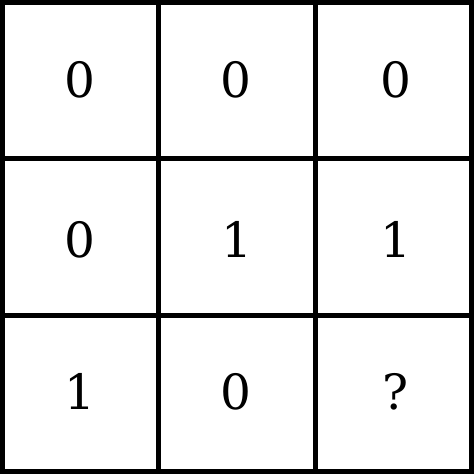}
		\caption{Optimal classical strategy for the magic square game.}
		\label{fig:ms-classical}
	\end{subfigure}
	\hfill
	\begin{subfigure}[b]{0.4\textwidth}
		\centering
		\includegraphics{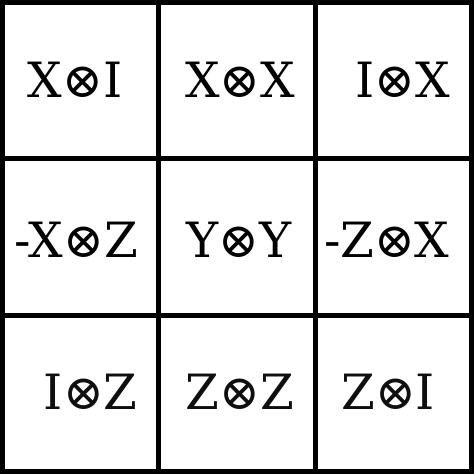}
		\caption{Perfect quantum strategy for the magic square game.}
		\label{fig:ms-quantum}
	\end{subfigure}
	\caption{Classical and quantum strategies for the magic square game.\\
		\subref{fig:ms-classical} An optimal classical strategy for the magic square game that allows Alice and Bob to win the game for 8 of the 9 possible questions. Filling the bottom right square consistently with the parity constraints for the rows (even) and columns (odd) is impossible. 
		\subref{fig:ms-quantum} A perfect quantum strategy defined by measuring the observables in the cells on two maximally entangled states. Note that the observables along each row and column commute pairwise.}
	\label{fig:magic-square-game}
\end{figure}

\appendix 

\section{Preliminaries}\label{sec:prelim}
\subsection{Non-local Games}\label{sec:non-local}

A two-player non-local game is a tuple $G=(\cX_1,\cX_2,\cY_1,\cY_2,P,W)$, where $\cX_i$ and $\cY_i$ are the question and answer sets for player $i=1,2$, respectively. 
The set $P\subset \cX_1 \times \cX_2$ is called the \textit{promise} of the game, the set of ``allowed'' questions.
The set $W\subset P \times \cY_1 \times \cY_2$ is the \textit{winning condition}, i.e., upon receiving the questions $(x_1,x_2)\in P$ and answering with $y_i\in \cY_i$, the players win the game if $(x_1,x_2,y_1,y_2)\in W$, and lose otherwise.

Every game $G=(\cX_1,\cX_2,\cY_1,\cY_2,P,W)$ with promise $P\subsetneq \cX_1 \times \cX_2$ can be turned into a \textit{promise-free} game $G'=(\cX_1,\cX_2,\cY_1,\cY_2,P',W')$ by declaring $P'= \cX_1 \times \cX_2$ and $W' = W \cup ( P^c \times \cY_1 \times \cY_2)$, i.e., the players win automatically if they receive a question pair outside the promise $P$.
We write $G=(\cX_1,\cX_2,\cY_1,\cY_2,W)$ for a promise-free game.

While the two players Alice and Bob can agree on a strategy beforehand, they are not allowed to communicate during the game.
A \textit{deterministic strategy} for Alice and Bob is a pair of deterministic functions $f_i\colon \cX_i \to \cY_i$ for $i=1,2$. 
A \textit{probabilistic strategy} for Alice and Bob is a probabilistic mixture of deterministic strategies.
For a given probability distribution $\pi\colon P\to [0,1]$ on the promised question set $P \subset \cX_1 \times \cX_2$, we define $\omega(G,\pi)$ as the maximal winning probability using probabilistic strategies.
If the given probability distribution on $P$ is the uniform distribution $\pi_U$, we use the shorthand $\omega_U(G) \equiv \omega(G,\pi_U)$. 
Note that the maximal winning probability is always achieved on an extremal point, i.e., a deterministic strategy.
A strategy achieving $\omega(G,\pi)=1$ is called \textit{perfect}.
	
\subsection{Capacity Region of Multiple Access Channels}\label{sec:macs}

\begin{figure}[ht]
	\centering
	\begin{tikzpicture}
	\draw (0,2) rectangle (2,-1) node[pos=.5] {$N(z|a,b)$};
	\draw (2,.5) -- (3,.5) node[pos=.5, above right] {$Z$};
	\draw (-1,1.5) -- (0,1.5) node[pos=0, above right] {$A$};
	\draw (-1,-.5) -- (0,-.5) node[pos=0, below right] {$B$};
	\end{tikzpicture}
	\caption{Mathematical model of a multiple access channel.\\
		The two senders and the single receiver of a multiple access channel $N(z|a,b)$ are associated with random variable $A$, $B$ and $Z$, respectively.}
	\label{fig:MAC}
\end{figure}

In this paper we only consider discrete memoryless multiple access channels without feedback, to which we refer simply as multiple access channels (MAC).
A two-sender MAC is a tuple $(\cA,\cB,\cZ,N(z|a,b))$, where $\cA$ and $\cB$ are the input alphabets for sender 1 and 2, respectively, $\cZ$ is the output alphabet for the single receiver, and $N(z|a,b)$ is a probability distribution for all pairs $(a,b)\in\cA\times\cB$ (see \Cref{fig:MAC}).

The following discussion is taken from \cite{EGK11}.
An $(R_1^{(n)},R_2^{(n)},n)$-code is a tuple $(\cM_1,\cM_2,a^n,b^n,\hat{z}^n)$, where:
\begin{itemize}
\item $\cM_1$ and $\cM_2$ are message sets with $|\cM_i|=2^{nR_i^{(n)}}$ for $i=1,2$;
\item $a^n\colon \cM_1 \to \cA^n$ and $b^n\colon \cM_2 \to \cB^n$ are encoding functions;
\item $\hat{z}^n\colon \cZ^n\to \cM_1\times \cM_2\cup\lbrace e\rbrace$ is a decoding function, with $e$ an arbitrary error message.
\end{itemize}
Without loss of generality we assume a uniform distribution over the messages $(M_1,M_2)$; in particular, the codewords $a^n(M_1)$ and $b^n(M_2)$ are independent.
The average probability of error is defined as
\begin{align}
P_e^{(n)} \coloneqq \prob\lbrace (\hat{M}_1,\hat{M}_2) \neq (M_1,M_2) \rbrace,
\end{align}
where $(\hat{M}_1,\hat{M}_2)$ is a random variable with probability mass function (pmf) $\hat{z}^n(N^{\times n}(z^n|a^n(M_1),b^n(M_2)))$.
A rate pair $(R_1,R_2)$ is said to be achievable if there exists a sequence of codes $\lbrace (R_1^{(n)},R_2^{(n)},n)\rbrace_{n\in\mathbb{N}}$ such that $\liminf_{n\to\infty} R_i^{(n)} = R_i$ and $\lim_{n\to\infty} P_e^{(n)} = 0$.
The capacity region $\cC(N)$ of the MAC $(\cA,\cB,\cZ,N(z|a,b))$ is the closure of the set of all achievable rate pairs $(R_1,R_2)$.
We also consider the sum-capacity $\Csum(N)$ defined as $\Csum(N)\coloneqq \sup\lbrace R_1 + R_2\colon (R_1,R_2)\in \cC(N)\rbrace$.

The capacity region $\cC(N)$ of a two-sender MAC $(\cA,\cB,\cZ,N(z|a,b))$ has a single-letter characterization \cite{Ahlswede71,Liao72}.
Let $(A,B)$ be a pair of discrete random variables jointly distributed according to the pmf $p_{A} p_{B}$, and let $Z$ be the channel output random variable with conditional pmf $N(z|a,b)$.
Define $\cR(A,B)$ as the set of rate pairs $(R_1,R_2)$ satisfying
\begin{align}
\begin{aligned}
R_1 &\leq I(A; Z|B)\\
R_2 &\leq I(B; Z|A)\\
R_1 + R_2 &\leq I(A,B;Z).
\end{aligned}
\label{eq:cap-region-constraints}
\end{align}
Then the capacity region $\cC(N)$ of the MAC $(\cA,\cB,\cZ,N(z|a,b))$ is the convex hull of the union of the regions $\cR(A,B)$ over all product distributions $p_{A} p_{B}$:
\begin{align}
\cC(N) = \operatorname{conv}\left( \bigcup\lbrace \cR(A,B)\colon (A,B)\sim p_{A}p_{B}\rbrace \right)
\label{eq:MAC-capacity-region}
\end{align}

\subsection{MACs and Entanglement Assistance}\label{sec:EA-MAC}

\begin{figure}
	\centering
	\begin{tikzpicture}
	\draw (0,-.25) -- (1,-.25) node[pos=0,above right] {$A_1$};
	\draw (0,-2.75) -- (1,-2.75) node[pos=0,above right] {$B_1$};
	\draw (1,-.25) -- (1.5,.25) -- (5.5,.25);
	\draw (1,-.25) -- (1.5,-.75) -- (2,-.75);
	\draw (1,-2.75) -- (1.5,-2.25) -- (2,-2.25);
	\draw (1,-2.75) -- (1.5,-3.25) -- (5.5,-3.25);
	\draw (2,-.5) rectangle (4.5,-2.5) node[pos=.5] {$P(a_2,b_2|a_1,b_1)$};
	\draw (4.5,-.75) -- (5.5,-.75) node[pos=0,above right] {$A_2$};
	\draw (4.5,-2.25) -- (5.5,-2.25) node[pos=0,above right] {$B_2$};
	\draw (5.5,-1) rectangle (8,.5) node[pos=0.5] {$f_1(a|a_1,a_2)$};
	\draw (5.5,-2) rectangle (8,-3.5) node[pos=0.5] {$f_2(b|b_1,b_2)$};
	\draw (8,-.25) -- (9.5,-.25) node[pos=0,above right] {$A$};
	\draw (8,-2.75) -- (9.5,-2.75) node[pos=0,above right] {$B$};
	\draw (9.5,0) rectangle (11.5,-3) node[pos=0.5] {$N(z|a,b)$};
	\draw (11.5,-1.5) -- (12.5,-1.5) node[pos=0.5,above right] {$Z$};
	\draw[dashed,thick,color=blue] (0.75,1) rectangle (8.75,-4);
	\node[color=blue] at (4.75,-4.5) {$E(a,b|a_1,b_1)$};
\end{tikzpicture}
	\caption{Entanglement-assisted coding scenario for a multiple access channel.\\
		The encoding $E$ (blue rectangle) is composed of the following: a correlation $P$ obtained from each sender measuring her system of a shared entangled state with POVMs selected according to the inputs $A_1$ and $B_1$ (see \eqref{eq:entangled-correlation}), and a post-processing of the outcomes $a_2$ and $b_2$ together with $a_1,b_1$ using functions $f_i$ to obtain the inputs $a$ and $b$ to the MAC $N$. If $A_1$ and $B_1$ are independent random variables, then the total channel $N\circ E$ can be interpreted as a MAC with input $A_1$ and $B_1$ and output $Z$, whose capacity region $\cC(N\circ E)$ (as defined in \eqref{eq:MAC-capacity-region}) is equal to the $d$-entanglement-assisted achievable rate region $\CeaO[d] (N)$ defined in \eqref{eq:EA-MAC-capacity-region}.}
	\label{fig:EA-MAC}
\end{figure}

In this paper we consider coding strategies for classical MACs assisted by entanglement shared between the two senders.\footnote{We briefly discuss entanglement assistance where each sender shares entanglement with the receiver at the end of \Cref{sec:magic-square-mac}.}
To formalize this setting, let $(\cA,\cB,\cZ,N(z|a,b))$ be a MAC as defined in \Cref{sec:macs}, and let the two senders $A$ and $B$ share an entangled state $|\psi\rangle_{S_AS_B}\in \mathbb{C}^d\ox \mathbb{C}^d$, where the $d$-dimensional quantum systems $S_A$ and $S_B$ with $|S_A|=|S_B|=d$ are with senders $A$ and $B$, respectively.
We then consider the following coding scenario.

Let $A_1$ and $B_1$ be random variables taking values in finite alphabets $\cA_1$ and $\cB_1$ for sender $A$ and $B$, respectively.
Depending on the value $a_1\in\cA_1$ of $A_1$, the first sender selects a positive operator-valued measure (POVM) $L^{(a_1)} = \lbrace L^{(a_1)}_{a_2}\rbrace_{a_2\in\cA_2}$ with 
\begin{align}
L^{(a_1)}_{a_2} &\geq 0 \quad \text{for all $a_1\in\cA_1$, $a_2\in\cA_2$ and}\\
\sum_{a_2\in\cA_2} L^{(a_1)}_{a_2} &= I_d \quad\text{for all $a_1\in\cA_1$.}
\end{align}
Here, $\cA_2$ is some finite alphabet, and $I_d$ denotes the identity operator on $\mathbb{C}^d$.
Likewise, for $b_1\in\cA_1$ the second sender selects a POVM $M^{(b_1)} = \lbrace M^{(b_1)}_{b_2}\rbrace_{b_2\in\cB_2}$ for some finite alphabet $\cB_2$ satisfying $M^{(b_1)}_{b_2}\geq 0$ for all $b_1\in\cB_1$, $b_2\in\cB_2$ and $\sum_{b_2\in\cB_2} M^{(b_1)}_{b_2} = I_d$ for all $b_1\in\cB_1$. 
Upon drawing $(a_1,b_1)$, the senders measure their respective half of the entangled state $|\psi\rangle_{S_AS_B}$ using the measurements $L^{(a_1)}$ and $M^{(b_1)}$, producing a correlation
\begin{align}
P(a_2,b_2|a_1,b_1) = \langle \psi| L^{(a_1)}_{a_2}\otimes M^{(b_1)}_{b_2} | \psi\rangle.
\label{eq:entangled-correlation}
\end{align}
Finally, the senders can each post-process their measurement outcomes $a_2$ and $b_2$ together with their inputs $a_1$ and $b_1$ to produce inputs $a$ and $b$ to the MAC $N$, which we summarize in a function $f_1(a|a_1,a_2) f_2(b|b_1,b_2)$.\footnote{Note that in principle this post-processing can be made part of the measurements with a potential increase of the local dimension $d$. However, we choose to keep them separate in order to link the local dimension $d$ to the non-local games considered in \Cref{sec:linear-system-games} in a clean way.}
In total, we have the classical channel 
\begin{align}
E(a,b|a_1,b_1) = f_1(a|a_1,a_2) f_2(b|b_1, b_2) P(a_2,b_2|a_1,b_1),
\label{eq:EA-encoding}
\end{align}
where the correlation $P$ is obtained from measuring the shared entangled state $|\psi\rangle_{S_AS_B}$ through \eqref{eq:entangled-correlation}.
The setup is depicted in \Cref{fig:EA-MAC}.

If we require the senders to draw $a_1$ and $b_1$ independently from a product distribution $p_{A_1}(a_1)p_{B_1}(b_1)$, then the channel $N\circ E$ with $E$ as defined in \eqref{eq:EA-encoding} can again be interpreted as a MAC with input alphabet $\cA_1\times\cB_1$ and output alphabet $\cZ$.
This prompts us to define the \textit{$d$-entanglement-assisted achievable rate region} of a classical MAC $N$ as
\begin{align}
\CeaO[d] (N) \coloneqq \bigcup_E \left\lbrace \cC(N\circ E)\right\rbrace,
\label{eq:EA-MAC-capacity-region}
\end{align}
where $\cC(\cdot)$ is the capacity region of an ordinary MAC defined in \eqref{eq:MAC-capacity-region}, and the union is over all classical channels $E$ as in \eqref{eq:EA-encoding} defined in terms of the following data:
\begin{itemize}
	\item an entangled state $|\psi\rangle_{S_AS_B}\in\mathbb{C}^d\ox \mathbb{C}^d$;
	\item arbitrary finite alphabets $\cA_1,\cB_1$ and $\cA_2,\cB_2$;
	\item POVMs $L^{(a_1)} = \left\lbrace L^{(a_1)}_{a_2}\right\rbrace_{a_2\in\cA_2}$ for $a_1\in\cA_1$ and $M^{(b_1)} = \left\lbrace M^{(b_1)}_{b_2}\right\rbrace_{b_2\in\cB_2}$ for $b_1\in\cB_1$, defined on $\mathbb{C}^d$;
	\item post-processings $f_1\colon \cA_1\times \cA_2 \ni (a_1,a_2)\mapsto a \in \cA$ and $f_2\colon \cB_1\times \cB_2 \ni (b_1,b_2)\mapsto b \in \cB$
\end{itemize} 
We also define the corresponding \textit{achievable sum rate}
\begin{align}
\CsumEA[d](N)\coloneqq \sup\lbrace R_1+R_2\colon (R_1,R_2)\in \CeaO[d](N)\rbrace.
\end{align}

The coding theorem \eqref{eq:MAC-capacity-region} for unassisted MACs implies that the region $\CeaO[d](N)$ in \eqref{eq:EA-MAC-capacity-region} is achievable, and hence a natural inner bound on the true entanglement-assisted capacity region of a MAC.
We expect that the latter is given by the regularized formula
\begin{align}
\Cea[d] (N) = \overline{\bigcup_{n\in\mathbb{N}} \frac{1}{n} \CeaO[d](N^{\times n})},
\end{align}
where $\overline{X}$ denotes the closure of a set $X$.
For the developments in \Cref{sec:undecidability}, we also define the \textit{entanglement-assisted achievable rate region}
\begin{align}
\CeaO (N) \coloneqq \bigcup_{d\in\mathbb{N}} \CeaO[d](N),
\end{align}
which is achievable by the two senders sharing entanglement on quantum systems of arbitrarily large but finite local dimension.

\section{Encoding a Non-local Game in a MAC}\label{sec:non-local-mac}

The following construction of a classical multiple access channel in terms of a non-local game is our main object of study.
It is inspired by a similar construction of an interference channel (two senders, two receivers) in terms of the CHSH game in \cite{QuekShor}, and can also be found in unpublished work by N\"{o}tzel and Winter \cite{NW17}, a portion of which has appeared in \cite{Noe19}.
Given a promise-free\footnote{Note that we can turn any game with promise into a promise-free one, as explained in \Cref{sec:non-local}.} non-local game $G = (\cX_1,\cX_2,\cY_1,\cY_2,W)$, we define the classical MAC $N_G \colon (\cX_1 \times \cY_1) \times (\cX_2 \times \cY_2) \rightarrow \cX_1 \times \cX_2$ as
\begin{align}
N_G(\hat{x}_1,\hat{x}_2 | x_1,y_1; x_2,y_2) &\coloneqq \begin{cases}
\delta_{x_1\hat{x}_1}\delta_{x_2 \hat{x}_2} & \text{if }(x_1,x_2,y_1,y_2) \in W,\\
(|\cX_1| |\cX_2|)^{-1} & \text{else}.
\end{cases}
\label{eq:non-local-mac}
\end{align}

In the above construction, to each player of the game we associate a sender in the MAC scenario with input alphabet $\cX_i\times \cY_i$ for $i=1,2$.
If the two senders input a question-answer tuple $(x_1,x_2,y_1,y_2)\in W$ that wins the non-local game $G$, the channel outputs the question pair $(x_1,x_2)$; otherwise, the channel outputs a question pair drawn uniformly at random.
In the following, for $i=1,2$ we denote by $X_i\sim \pi_{X_i}$ the random variables corresponding to the questions for Alice and Bob, by $Y_i\sim p_{Y_i|X_i}\pi_{X_i}$ the random variables corresponding to the answers, and by $Z$ the random variable corresponding to the output of the channel $N_G$ defined in \eqref{eq:non-local-mac} taking values in $\cX_1\times \cX_2$.

As discussed in \Cref{sec:macs}, the capacity region of a MAC is computed in terms of a product probability distribution $p_{X_1Y_1}(x_1,y_1)p_{X_2Y_2}(x_2,y_2)$ on the set of inputs to $N$.
For the MAC \eqref{eq:non-local-mac}, we can think of this input distribution in the following way: Given a product probability distribution $\pi(x_1,x_2) = \pi_{X_1}(x_1) \pi_{X_2}(x_2)$ on the question set, the players produce answers $y_i$ to the game according to the probabilistic strategy $p_{Y_1|X_1}(y_1|x_1) p_{Y_2|X_2}(y_2|x_2)$ on which they agreed prior to starting the game.
This allows us to connect the sum rate capacity of the channel $N_G$ to the winning probability $\omega(G,\pi)$ as follows:

\begin{prop}\label{prop:sum-capacity}
	Let $G=(\cX_1,\cX_2,\cY_1,\cY_2,W)$ be a promise-free non-local game, $\pi(x_1,x_2) = \pi_{X_1}(x_1) \pi_{X_2}(x_2)$ a probability distribution on the questions set $\cX_1\times \cX_2$, and $p_{Y_1|X_1}(y_1|x_1) p_{Y_2|X_2}(y_2|x_2)$ a probabilistic strategy for Alice and Bob.
	For the MAC $N_G$ defined in terms of $G$ according to \eqref{eq:non-local-mac}, let $X_i, Y_i, Z$ be the random variables corresponding to the questions, answers, and channel output, respectively, as described above.
	We then have
	\begin{align}
	I(X_1Y_1X_2Y_2; Z) = H(Z) - p_L (\log |\cX_1| + \log |\cX_2|),
	\end{align}
	where $p_L = \sum_{(x_1,y_1,x_2,y_2)\notin W} \Pr(x_1,y_1,x_2,y_2)$ denotes the losing probability given the distribution $\pi(x_1,x_2)$ on the questions set and the probabilistic strategy $p_{Y_1|X_1}(y_1|x_1) p_{Y_2|X_2}(y_2|x_2)$.
\end{prop}

\begin{proof}
	We first expand the mutual information $I(X_1Y_1X_2Y_2; Z)$ as
	\begin{align}
	I(X_1Y_1X_2Y_2; Z) = H(Z) - H(Z|X_1 Y_1 X_2 Y_2).
	\end{align}
	Setting $d=|\cX_1||\cX_2|$ and recalling that $W\subset \cX_1\times \cY_1 \times \cX_2 \times \cY_2$ is the winning set for $G$, the conditional entropy can be expressed as
	\begin{align}
	H(Z|X_1 Y_1 X_2 Y_2) &= \sum_{x_1,y_1,x_2,y_2} \prob(x_1,y_1,x_2,y_2) H(Z|X_1=x_1,Y_1=y_1,X_2=x_2,Y_2=y_2)\\
	&= \sum_{(x_1,y_1,x_2,y_2)\notin W} \prob(x_1,y_1,x_2,y_2) H(Z|X_1=x_1,Y_1=y_1,X_2=x_2,Y_2=y_2)\\
	&= \log d\, \sum_{(x_1,y_1,x_2,y_2)\notin W} \prob(x_1,y_1,x_2,y_2)\\
	&= p_L \log d,
	\end{align}
	where the second equality follows since for $(x_1,y_1,x_2,y_2)\in W$ the channel $N_G$ outputs $(x_1,x_2)$ deterministically, and hence $H(Z|X_1=x_1,Y_1=y_1,X_2=x_2,Y_2=y_2)=0$ in this case.
\end{proof}

\begin{prop}\label{prop:classical-bound}
	Let $G=(\cX_1,\cY_1,\cX_2,\cY_2,W)$ be a promise-free non-local game with $\omega_U(G) < 1$, and consider the MAC $N_G$ defined as in \eqref{eq:non-local-mac}.
	Using the same notation as in \Cref{prop:sum-capacity}, for all $0<\delta < -\log \omega_U(G)$ there exists an $\eps > 0$ such that
	\begin{align}
		I(X_1 Y_1 X_2 Y_2;Z) \leq \max\lbrace (1-\eps) (\log |\cX_1| + \log |\cX_2|), \log |\cX_1| + \log |\cX_2| - \delta\rbrace.\label{eq:capacity-upper-bound}
	\end{align} 
	For a given $\delta>0$ the maximal value of $\eps$ is given by the $\eps^*$ satisfying
	\begin{align}
	\frac{\delta+h(\eps^*)}{1-\eps^*} = \delta(\eps^*\|1-\omega_U(G)),\label{eq:eps-from-delta}
	\end{align}
	where $h(x) = -x \log{x} - (1-x) \log{(1-x)}$ denotes the binary entropy and $\delta(x\|y)\coloneqq x\log\frac{x}{y} + (1-x)\log\frac{1-x}{1-y}$ denotes the binary relative entropy.
\end{prop}

The strategy of the proof of \Cref{prop:classical-bound} is the following: the goal of the proposition is to provide an upper bound on the sum rate capacity that separates it from the maximal value $\log|\cX_1| + \log |\cX_2|$.
By the formula given in \Cref{prop:sum-capacity}, the maximal value $\log|\cX_1| + \log |\cX_2|$ is attained if and only if $p_L$ vanishes and $H(Z)$ attains its maximal value.
We therefore need to show that we cannot have $p_L \approx 0$ and $H(Z)\approx \log|\cX_1| + \log |\cX_2|$ at the same time.

\begin{proof}[Proof of \Cref{prop:classical-bound}]
	We again set $d=|\cX_1||\cX_2|$.
	For the purpose of bounding the sum rate capacity $I(X_1 Y_1 X_2 Y_2;Z)$ away from the maximal value $\log d$, we can assume without loss of generality that the losing probability $p_L=1-\omega(G,\pi)$ for Alice and Bob is strictly positive, $p_L>0$:
	In case $p_L=0$, the probability distribution $\pi$ on $\cX_1\times \cX_2$ necessarily has support $\supp \pi$ strictly contained in $\cX_1\times \cX_2$, as by assumption the game $G$ cannot be won with certainty on receiving any one of the full set $\cX_1\times \cX_2$ of questions.
	Hence, 
	\begin{align}
	I(X_1 Y_1 X_2 Y_2;Z) \leq \log |\supp \pi| \leq \log{(d-1)} < \log d
	\end{align}
	in this case, as Alice and Bob have to lose on at least one question pair.
	Furthermore, we can assume w.l.o.g.~that $p_L \leq 1-\omega_U(G)$, since $p_L>1-\omega_U(G)$ and $\omega_U(G)<1$ imply that 
	\begin{align}
	I(X_1 Y_1 X_2 Y_2;Z) < \omega_U(G) \log d < \log d.
	\end{align}
	Therefore, we may assume that $0 < p_L \leq 1-\omega_U(G)$ for the remainder of the proof.
	
	We prove the assertion of the theorem by contradiction.
	To this end, assume that
	\begin{align}
	H(Z) &\geq \log d - \delta \label{eq:entropy-assumption}
	\end{align}
	for some $0 < \delta < -\log \omega_U(G)$.
	Define a random variable $W$ by
	\begin{align}
	W = \begin{cases}
	1 & \text{Alice and Bob win the game;}\\
	0 & \text{Alice and Bob lose the game,}
	\end{cases}
	\label{eq:W}
	\end{align}
	taking values $1$ and $0$ with probability $1-p_L$ and $p_L$, respectively.
	By the non-negativity of conditional entropy and \eqref{eq:entropy-assumption} we have
	\begin{align}
	H(W) + H(Z|W)  = H(ZW)\geq H(Z) \geq \log d - \delta.\label{eq:conditional-entropy-nonneg}
	\end{align}
	Expanding the left-hand side of \eqref{eq:conditional-entropy-nonneg} gives
	\begin{align}
	h(p_L) + (1-p_L) H(X_1 X_2) + p_L \log d \geq \log d - \delta,
	\end{align}
	which can be rearranged to
	\begin{align}
	H(X_1X_2) &\geq \log d - \gamma
	\end{align}
	with $\gamma \coloneqq \frac{\delta+ h(p_L)}{1-p_L}$.
	Observe that $D( \pi_{X_1} \pi_{X_2} \| \pi_U) = \log d - H(X_1 X_2)$, and hence
	\begin{align}
	\gamma \geq D( \pi_{X_1} \pi_{X_2} \| \pi_U).\label{eq:gamma-bound}
	\end{align}
	
	Let now $Q = Q_{Y_1|X_1} Q_{Y_2|X_2}$ be the optimal probabilistic strategy for $G$ given the distribution $\pi_{X_1} \pi_{X_2}$ on the questions, and denote by $q_L$ the losing probability of the same strategy $Q$ with questions drawn uniformly at random. 
	Furthermore, let $\chi_W\colon \cX_1 \times \cY_1 \times \cX_2 \times \cY_2\to \lbrace 0,1\rbrace$ be the characteristic function of the winning set $W\subset  \cX_1 \times \cY_1 \times \cX_2 \times \cY_2$.
	Applying the data-processing inequality with respect to $\chi_W\circ Q$ to \eqref{eq:gamma-bound}, we obtain
	\begin{align}
	\gamma \geq \delta(p_L \| q_L) \geq \delta(p_L \|1-\omega_U(G) ),\label{eq:gamma-inequality}
	\end{align}
	where $\delta(x\|y)\coloneqq x\log\frac{x}{y} + (1-x)\log\frac{1-x}{1-y}$ denotes the binary relative entropy, and the second inequality follows from the monotonicity of $y\mapsto \delta(x\|y)$ for $y\geq x$ and the fact that $q_L \geq 1-\omega_U(G)$.	
	
	The function $\gamma(x)=\frac{\delta+h(x)}{1-x}$ is monotonically increasing for all $x>0$, and $\lim_{x\to 0} \gamma(x) = \delta < -\log \omega_U(G)$ by assumption.
	On the other hand, $x\mapsto \delta(x\|1-\omega_U(G))$ is monotonically decreasing for $x\in[0,1-\omega_U(G)]$, and $\lim_{x\to 0} \delta(x \|1-\omega_U(G) ) = -\log \omega_U(G)$.
	Hence, there exists an $\eps>0$ such that \eqref{eq:gamma-inequality} is violated for all $p_L < \eps$, which means that we either have $p_L \geq \eps$ or the assumption \eqref{eq:entropy-assumption} is false.
	In the first case, 
	\begin{align}
	 	I(X_1 Y_1 X_2 Y_2;Z) = H(Z) - p_L \log d \leq (1-\eps) \log d,
	\end{align}
	while in the second case,
	\begin{align}
		I(X_1 Y_1 X_2 Y_2;Z) < (1-p_L)\log d-\delta < \log d - \delta.
	\end{align}
	By the arguments above, for a given $\delta>0$ the maximal value of $\eps$ is given by the $\eps^*$ satisfying
	\begin{align}
	\frac{\delta+h(\eps^*)}{1-\eps^*} = \delta(\eps^*\|1-\omega_U(G)),
	\end{align}
	which concludes the proof.
\end{proof}
\begin{rem}\label{rem:determine-optimal-delta}
	In applications of \Cref{prop:classical-bound}, the optimal (minimal) upper bound in \Cref{prop:classical-bound} can be obtained by optimizing the right-hand side of \eqref{eq:capacity-upper-bound} over $\delta\in (0,-\log\omega_U(G))$ and computing $\eps$ via \eqref{eq:eps-from-delta}.
\end{rem}
	
\section{Magic Square Game}\label{sec:magic-square-game}

Consider a $3\times 3$-matrix whose rows and columns are labeled by $r,c\in \lbrace 0,1,2\rbrace$, respectively.
The magic square game $G_{\MS}$ \cite{Mer90,Peres1990,Ara02b,BBT05} is a two-player game in which Alice and Bob receive questions $r,c\in \lbrace 0,1,2\rbrace$ respectively, labeling a row $r$ for Alice and a column $c$ for Bob.
They answer with 3-bit strings $s,t\in \lbrace 0,1\rbrace^3$, where the bits in $s,t$ correspond to the cells in $r,c$, respectively.
Alice and Bob win the game if the following three conditions are satisfied:
\begin{enumerate}
	\item the parity of Alice's bit string $s$ is even: $s_0\oplus s_1 \oplus s_2 = 0$;
	\item the parity of Bob's bit string $t$ is odd: $t_0 \oplus t_1 \oplus t_2 = 1$;
	\item the bit strings agree in the overlapping cell $(r,c)$: $s_c = t_r$.
\end{enumerate}

\subsection{Classical Strategies}\label{sec:magic-square-game-classical}
The two parity constraints for Alice's and Bob's bit strings $s$ and $t$ render any deterministic perfect classical strategy for $G_{\MS}$ impossible, since the latter corresponds to a fixed valid filling of the nine cells of the magic square with bits such that conditions 1-3 above are satisfied.
However, according to condition 1 the parity of all cells is even, while according to condition 2 this parity should be odd.

If the questions $(r,c)$ are drawn uniformly at random, the best deterministic strategy for Alice and Bob consists in filling 8 of the 9 cells with valid bits.
Hence, the optimal deterministic strategy has winning probability $8/9$, and in fact $\omega_U(G_{\MS})=8/9$ \cite{BBT05}.

\subsection{A Perfect Quantum Strategy}\label{sec:magic-square-game-quantum}
\textcite{BBT05} described the following perfect quantum strategy for the magic square game $G_{\MS}$ that is equivalent to the commuting observables strategy devised by \textcite{Mer90} and \textcite{Peres1990} in \Cref{fig:ms-quantum}:
Consider the 4-qubit entangled state
\begin{align}
|\psi\rangle_{A_1A_2B_1B_2} = \frac{1}{2} \left( |00\rangle_{A_1A_2}|11\rangle_{B_1B_2} + |11\rangle_{A_1A_2}|00\rangle_{B_1B_2} - |01\rangle_{A_1A_2}|10\rangle_{B_1B_2} - |10\rangle_{A_1A_2}|01\rangle_{B_1B_2} \right),
\label{eq:psi}
\end{align}
where qubits $A_1A_2$ are with Alice, and $B_1B_2$ are with Bob.
Furthermore, consider the following 2-qubit unitaries:
\begin{align}
\begin{aligned}
U_0 &= \frac{1}{\sqrt{2}}\begin{pmatrix}
\phantom{-}\ii & \phantom{-}0 & \phantom{-}0 & \phantom{-}\ii\\ \phantom{-}0 & -\ii & \phantom{-}1 & \phantom{-}0 \\ \phantom{-}0 & \phantom{-}\ii & \phantom{-}1 & \phantom{-}0\\ \phantom{-}1 & \phantom{-}0 & \phantom{-}0 & \phantom{-}\ii
\end{pmatrix} &
U_1 &= \frac{1}{2} \begin{pmatrix}
\phantom{-}\ii & \phantom{-}1 & \phantom{-}1 & \phantom{-}\ii\\ -\ii & \phantom{-}1 & -1 & \phantom{-}\ii\\ \phantom{-}\ii & \phantom{-}1 & -1 & -\ii\\-\ii & \phantom{-}1 & \phantom{-}1 & -\ii
\end{pmatrix} &
U_2 &= \frac{1}{2}\begin{pmatrix}
-1 & -1 & -1 & \phantom{-}1\\ \phantom{-}1 & \phantom{-}1 & -1 & \phantom{-}1\\ \phantom{-}1 & -1 & \phantom{-}1 & \phantom{-}1\\ \phantom{-}1 & -1 & -1 & -1
\end{pmatrix} \\[0.5em]
V_0 &= \frac{1}{2} \begin{pmatrix}
\phantom{-}\ii & -\ii & \phantom{-}1 & \phantom{-}1\\ -\ii & -\ii & \phantom{-}1 & -1\\ \phantom{-}1 & \phantom{-}1 & -\ii & \phantom{-}\ii\\ -\ii & \phantom{-}\ii & \phantom{-}1 & \phantom{-}1
\end{pmatrix} &
V_1 &= \frac{1}{2} \begin{pmatrix}
-1 & \phantom{-}\ii & \phantom{-}1 & \phantom{-}\ii\\ \phantom{-}1 & \phantom{-}\ii & \phantom{-}1 & -\ii\\ \phantom{-}1 & -\ii & \phantom{-}1 & \phantom{-}\ii\\ -1 & -\ii & \phantom{-}1 & -\ii
\end{pmatrix} &
V_2 &= \frac{1}{\sqrt{2}} \begin{pmatrix}
\phantom{-}1 & \phantom{-}0 & \phantom{-}0 & \phantom{-}1\\ -1 & \phantom{-}0 & \phantom{-}0 & \phantom{-}1\\ \phantom{-}0 & \phantom{-}1 & \phantom{-}1 & \phantom{-}0\\ \phantom{-}0 & \phantom{-}1 & -1 & \phantom{-}0
\end{pmatrix}
\end{aligned}
\end{align}
Upon receiving the questions $(r,c)$, Alice applies the unitary $U_r$ to her qubits $A_1A_2$, while Bob applies $V_c$ to his qubits $B_1B_2$.
They each measure their respective qubits of the resulting state $U_r \ox V_c |\psi\rangle$ in the computational basis and obtain measurement outcomes $s_0s_1$ and $t_0t_1$.
As a last step, they complete their 2-bit outcome with a third bit such that the parity conditions of the magic square game are satisfied: Alice chooses $s_2$ such that $s_0\oplus s_1\oplus s_2=0$, while Bob chooses $t_2$ such that $t_0\oplus t_1\oplus t_2 = 1$.
A lengthy but straightforward computation shows that this strategy indeed produces a valid answer pair $(s,t)$ for every possible question pair $(r,c)$.

\subsection{MAC Based on the Magic Square Game}\label{sec:magic-square-mac}

Specializing definition \eqref{eq:non-local-mac} to the magic square game $G_{\MS}$ described above, we set $\cR=\lbrace 0,1,2\rbrace $, $\cC=\cR$, $\cS= \lbrace 0,1\rbrace^3$, $\cT= \cS$, and consider the following channel:
\begin{align}
\begin{aligned}
\Nms \colon (\cR \times \cS) \times (\cC \times \cT) &\longrightarrow \cR \times \cS\\
\Nms(\hat{r},\hat{s}|r,s;c,t) &\coloneqq \begin{cases}
\delta_{r\hat{r}} \delta_{s\hat{s}} & \text{if } (r,s; c,t) \in W,\\
\frac{1}{9} & \text{else},
\end{cases}
\end{aligned}
\label{eq:magic-square-MAC}
\end{align}
where $W\subset \cR\times\cS\times\cC\times\cT$ is the subset of instances $(r,s; c,t)$ winning the magic square game.

Using the perfect quantum strategy for the magic square game detailed in \Cref{sec:magic-square-game-quantum}, for any question pair $(r,c)$ Alice and Bob can produce answers $(s,t)$ such that $(r,s,c,t)\in W$.
Hence, with a uniform distribution over the questions $\cR\times \cC$ they can achieve the maximal sum rate of $\log 9 \approx 3.16993$ for the magic-square-MAC \eqref{eq:magic-square-MAC}.
To bound the sum rate achievable by classical strategies corresponding to product input distributions on $(\cR\times\cS)\times(\cC\times\cT)$, our goal is to find the smallest upper bound on $I(RSCT;Z)$ given by \Cref{prop:classical-bound} (we again use capital Latin letters for the random variables corresponding to the question and answer sets, as well as $Z$ for the channel output random variable): 
\begin{align}
I(RSCT;Z) \leq \max\lbrace (1-\eps^*) \log 9, \log 9 - \delta\rbrace \eqqcolon u(\delta)
\label{eq:magic-square-upper-bound}
\end{align}
for some $\delta\in(0,\log \frac{9}{8})$ and the corresponding optimal $\eps^*$ determined through \eqref{eq:eps-from-delta}.
As explained in \Cref{rem:determine-optimal-delta}, we find the optimal $\delta^* = 0.03299$ (using, e.g., Mathematica), which yields $\eps(\delta^*) = 0.01040$ and $I(RSCT;Z) \leq u(\delta^*)= 3.13694$.
In \Cref{fig:magic-square-optimal-delta} we plot the upper bound \eqref{eq:magic-square-upper-bound} as a function of $\delta\in[0,\log 9/8]$.

\begin{figure}
	\centering
	\begin{tikzpicture}
	\begin{axis}[xlabel=$\delta$,ylabel=$u(\delta)$,xmin=0,xmax=0.169925,tick label style={/pgf/number format/fixed, /pgf/number format/precision=4},
	xtick = {0,0.025,0.05,0.075,0.1,0.125,0.15}
	]
	\addplot[thick,color=blue] table[x=d,y=u] {magic-square-ub.dat};
	\addplot[domain=0:0.169925,dashed,semithick,color=gray] {3.16993};
	\end{axis}
	\end{tikzpicture}
	\caption{Upper bounds on the classical sum rate for the multiple access channel based on the magic square game.\\
		Plotted is the upper bound $u(\delta)$ defined in \eqref{eq:magic-square-upper-bound} as a function of $\delta$, with $\eps=\eps(\delta)$ chosen maximally such that \eqref{eq:gamma-inequality} is violated.
		The minimum occurs at $\delta^* = 0.03299$ giving $\eps(\delta^*) = 0.01040$ and $u(\delta^*)= 3.13694$.}
	\label{fig:magic-square-optimal-delta}
\end{figure}

We can compare the upper bound $u(\delta^*)= 3.13694$ to a lower bound on the sum rate computed by numerically maximizing the mutual information $I(RSCT;Z)$ with respect to product probability distributions (see \eqref{eq:cap-region-constraints}).
Carrying out this optimization in MATLAB in repeated runs using different random starting points gives a lower bound of $2.84195$ on the true maximum.
Assuming that this value is close to the true maximum, this result suggests that our upper bound $u(\delta^*)= 3.13694$ on the sum rate can likely be further improved.
We also computed an inner bound on the capacity region $\cC$ of the MAC \eqref{eq:magic-square-MAC} using the method detailed in \cite[Sec.~II.A]{CPFV10}.
This inner bound on the capacity region is plotted in \Cref{fig:magic-square-cap-region}.

\begin{figure}[h]
	\centering
	\begin{tikzpicture}
	\begin{axis}[xlabel=$R_1$,ylabel=$R_2$,ymax=1.8,xmax=1.8,xmin = 0,ymin = 0,
	legend style = {at = {(0,1.1)},anchor = south west,column sep = 0.2cm}, legend cell align = left,]
	\addplot[thick, color=blue] table[x=r1,y=r2] {cap_region.dat};
	\addplot[thick,dashed,color=red,domain=0:1.1539] {1.68805};
	\addplot[thick,dashdotted,color=cyan,domain=1:2] {3.13694-x};
	\fill(axis cs:1.585,1.585) circle[color=green, radius=1];
	\addplot[thick,dashed,color=red] coordinates {(1.68805, 1.1539) (1.68805,0)};
	\addplot[thick,dashed,color=red] coordinates {(1.1539, 1.68805) (1.68805,1.1539)};
	\legend{Inner bound on $\cC(\Nms)$,Constraints on $\cC(\Nms)$ given by \eqref{eq:cap-region-constraints},Upper bound $u(\delta^*)$};
	\end{axis}
	\end{tikzpicture}
	\caption{Inner and outer bounds on the capacity region of the multiple access channel based on the magic square game.\\
		The inner bound on the capacity region $\cC(\Nms)$ of the MAC \eqref{eq:magic-square-MAC} based on the magic square game is shown in solid blue.
		Approximate values of the outer pentagonal bound on $\cC$ given by optimizing the individual constraints in \eqref{eq:cap-region-constraints} for $R_1$, $R_2$ and $R_1+R_2$ are marked by dashed red lines.
		The dash-dotted cyan line is the (optimized) upper bound on the sum rate from \Cref{prop:classical-bound}.
	The black dot is the rate pair $(\log 3,\log 3)$ achievable by the entanglement-assisted coding strategy explained in \Cref{sec:magic-square-game-quantum}.}
	\label{fig:magic-square-cap-region}
\end{figure}

We briefly comment on a different type of entanglement assistance for a MAC where each sender shares entanglement with the receiver.
This communication scenario was discussed by \textcite{HDW08} for quantum multiple access channels $\cN\colon A'B'\to C$ mapping quantum systems $A'$ in Alice's possession and $B'$ in Bob's possession to a quantum system $C$ in possession of the receiver Charlie.
In addition to entanglement assistance, quantum MACs have been studied in various other scenarios \cite{Winter01,YHD08,horodecki09,Qi2018,Sen18a}.

The following capacity region for entanglement-assisted quantum MACs is proved in \cite{HDW08}:
Let $|\phi\rangle_{AA'}$ and $|\psi\rangle_{BB'}$ be pure quantum states, and set $\omega_{ABC} = (\id_A \ox \id_B \ox \cN_{A'B'\to C})(\phi_{AA'}\ox\psi_{BB'})$.
Let $\cC_E(\cN,\phi,\psi)$ be the set of all non-negative rate pairs $(R_1,R_2)$ satisfying
\begin{align}
R_1 &\leq I(A;C|B)\\
R_2 &\leq I(B;C|A)\\
R_1 + R_2 &\leq I(AB;C),
\end{align}
where the quantum (conditional) mutual informations on the right-hand sides are evaluated on the state $\omega_{ABC}$.
Define $\widetilde{\cC}_E(\cN)$ as the union over all states $\phi$ and $\psi$.
Then the entanglement-assisted capacity region $\cC_E(\cN)$ of a quantum MAC $\cN$ is equal to
\begin{align}
\cC_E(\cN) = \overline{\bigcup_{n\in\mathbb{N}}\frac{1}{n} \widetilde{\cC}_E(\cN^{\ox n}) }.
\end{align}
Moreover, we have the following single-letter upper bound on the sum rate:
\begin{align}
R_1 + R_2 \leq \max_{\phi_{AA'},\psi_{BB'}} I(AB;C).\label{eq:q-sum-rate-bound}
\end{align}

We now specialize the above entanglement-assisted setting to classical MACs $N\colon \cA\times\cB \to \cZ$ as introduced in \Cref{sec:macs}.
Any classical channel necessarily completely dephases a quantum system with respect to some fixed basis.
Hence, choosing bases $\lbrace |i\rangle_A\rbrace$ and $\lbrace |i\rangle_B\rbrace$ and fixing pure quantum states $\phi_{AA'}$ and $\psi_{BB'}$, the joint input state of Alice and Bob for a classical MAC is of the form
\begin{align}
\sum_{i,j} p_i p_j \phi_{A}^i \ox |i\rangle\langle i|_{A'} \ox \psi_{B}^j \ox |j\rangle\langle j|_{B'},
\label{eq:joint-input-state}
\end{align}
where $\lbrace p_i\rbrace$ with $p_i = \tr (|i\rangle\langle i|_{A'}\phi_{AA'})$ and $\lbrace p_j\rbrace$ with $p_j = \tr (|j\rangle\langle j|_{B'}\psi_{BB'})$ are probability distributions, and $\phi^i_A = \frac{1}{p_i} \langle i|\phi|i\rangle_{A'}$ and $\phi^j_A = \frac{1}{p_j} \langle j|\psi|j\rangle_{B'}$.
The MAC $N$ maps the joint input state in \eqref{eq:joint-input-state} to a state
\begin{align}
\omega_{ABZ} = \sum_{i,j} p_i p_j N(k|i,j) \phi_{A}^i  \ox \psi_{B}^j \ox |k\rangle\langle k|_Z.
\end{align}
On the other hand, consider a classical state
\begin{align}
\theta_{ABZ} = \sum_{i,j} p_i p_j N(k|i,j) |i\rangle\langle i|_{A} \ox |j\rangle\langle j|_{B} \ox |k\rangle\langle k|_Z,
\end{align}
and observe that $\omega_{ABZ}$ can be obtained from $\theta_{ABZ}$ by a quantum operation that first measures the (classical) systems $AB$ in $\theta$ and depending on the outcome $(i,j)$ prepares the state $\phi_{A}^i  \ox \psi_{B}^j$.
Hence, by the data processing inequality for the quantum mutual information we have
\begin{align}
I(AB;Z)_\omega \leq I(AB;Z)_\theta,
\label{eq:sum-rate-classical-upper-bound}
\end{align}
and $I(AB;Z)_\theta$ is the classical mutual information with respect to the product probability distribution $p_i p_j$ appearing in the sum rate constraint for the classical MAC $N$ given in \eqref{eq:cap-region-constraints}.

From the above discussion and \eqref{eq:q-sum-rate-bound}, we conclude that for a classical MAC entanglement shared between each sender and the receiver cannot increase the achievable sum rate $R_1+R_2$.
In contrast, we showed in this section that entanglement shared between the senders can indeed increase the sum rate up to the maximal value.

\section{Linear System Games}\label{sec:linear-system-games}

In this section we discuss non-local games $G_{\LS}$ based on linear systems of equations \cite{cleve2014characterization}.
Let $Ax=b$ be an $m\times n$ linear system of equations over $\mathbb{F}_2$.
We denote by $V_i=\lbrace j\in[n]\colon A_{ij} \neq 0\rbrace$ the indices of variables appearing the $i$-th equation of the linear system.
In the linear system game, Alice receives as a question an index $i\in [m]$ labeling a row in the linear system.
She replies with a vector $y\in\mathbb{F}_2^n$ of values for $x$ such that $\sum_{j\in V_i} y_j = b_i$.
Bob receives as a question an index $j\in[n]$, and he answers with a bit $x_j$ corresponding to an assignment of the variable $x_j$.
Alice and Bob win the game if either $j\notin V_i$ or $y_j = x_j$.

A linear system game $G_{\LS}$ defined in terms of a linear system $Ax=m$ can be associated with a certain finitely-presented group $\Gamma(A,b)$ called a solution group.
The maximal winning probability using quantum strategies can then be related to approximate representations of $\Gamma(A,b)$ \cite{Slofstra2019}.
\textcite{SV18} showed that suitable approximate representations of $\Gamma(A,b)$ (giving rise to near-perfect quantum strategies) do exist provided the dimension of the representation space, called the hyperlinear profile, is large enough.
They exhibited a particular example $G_{\SV}$ of a linear system game based on a suitable solution group $\Gamma(A,b)$, for which the above observations can be translated into lower and upper bounds on the local dimension $d$ of any quantum strategy for $G_{\SV}$.
In terms of the losing probability $p_L = 1-\omega_U(G_{\SV})$ and constants $C,C'$, the following bounds are proved in \cite{SV18}:
\begin{align}
\frac{C}{p_L^{1/6}} \leq d \leq \frac{C'}{p_L^{1/2}}.
\label{eq:bounds-app}
\end{align}

\subsection{Limiting the Entanglement Assistance}

\begin{prop}\label{prop:limited-entanglement-assistance}
If Alice and Bob are constrained to quantum strategies with dimension at most $d$, then the sum rate capacity of $N_{G_{\SV}}$ is bounded away from perfect, i.e., $\log{m} + \log{n}$, by $\Theta{(\frac{1}{d^{13}})}$.
\end{prop}

\begin{proof}
Let $G_{\SV}$ be the linear system game defined in \cite{SV18}.
By the discussion above and \eqref{eq:bounds-app}, we have the following lower bound for the losing probability if Alice and Bob only use $d$-dimensional quantum strategies:
\begin{align}
    1 - \omega_{U} (G_{\SV}) \geq \frac{C_1}{d^{6}},
\end{align}
for some constant $C_1>0$. In order to use \eqref{prop:classical-bound}, we let $\delta = \frac{C_1}{d^{13}}$ and assume that $\eps^* < \delta$. For large $d$, we can upper-bound the left-hand side of \cref{eq:eps-from-delta} by 
\begin{align}
\frac{\delta+h(\delta)}{1-\delta} \geq \delta(\eps^*\|1-\omega(G_{\SV})), 	
\end{align}
where we used $h(\eps^*) \leq h(\delta)$ whenever $\delta < \frac{1}{2}$. Next, observe that for $\delta \in [0,\frac{1}{2}]$ the binary entropy term $h(\delta)$ is upper-bounded by $a \delta^{\alpha}$ for $\alpha < 1$ and $a$ large enough. Letting $\alpha = \frac{25}{26}$, we can underestimate the right-hand side via Pinsker's inequality and get
\begin{align}
\delta(\eps^*\|1-\omega_{U}(G_{\SV})) \geq 
\frac{2}{\ln{2}} \Big[ \eps^{*} - (1-\omega_{U}(G_{\SV})) \Big]^{2} \geq 
\frac{2}{\ln{2}} \Big[ \eps^{*} - \frac{C_1}{d^{6}} \Big]^{2} \geq
\frac{2}{\ln{2}} \Big[ \frac{C_1^{2}}{d^{12}} - 2 \delta \frac{C_1}{d^{6}} \Big]
\end{align}
Putting it all together, we conclude the following inequality 
\begin{align}\label{ineq: converse}
\frac{\delta+ a \delta^{\alpha}}{1-\delta} \geq 
\frac{2C_1^{2}}{\ln{2}} \Big[ \frac{1}{d^{12}} - 2 \frac{1}{d^{19}} \Big].
\end{align}
Observe that as $d$ goes to infinity, the right-hand side goes as $1/d^{12}$, while the left-hand side goes as $1/d^{12.5}$. At some large enough $d$, \eqref{ineq: converse} is violated. $\eps^*$ cannot be smaller than $\delta$ for large enough $d$. Hence, by \eqref{prop:classical-bound} we have the following upper bound on the sum rate capacity:
\begin{align}
    I(X_1 Y_1 X_2 Y_2;Z) \leq \log{n} + \log{m} - \frac{C_1}{d^{13}}
\end{align}
Since $d$-dimensional quantum strategies subsume all lower dimensional strategies, this converse provides a limit, if implicit, on how well $N_{G_{\SV}}$ can be used for strategies with small dimension.
\end{proof}
\subsection{Achievable Strategies Using $d$-dimensional Maximally Entangled States}

\begin{figure}[ht]
\centering
\begin{tikzpicture}
\draw (0,-.25) -- (1,-.25) node[pos=0,above right] {$X_1$};
\draw (0,-2.75) -- (1,-2.75) node[pos=0,above right] {$X_2$};
\draw (1,-.25) -- (1.5,.25) -- (5.5,.25);
\draw (1,-.25) -- (1.5,-.75) -- (2,-.75);
\draw (1,-2.75) -- (1.5,-2.25) -- (2,-2.25);
\draw (1,-2.75) -- (1.5,-3.25) -- (5.5,-3.25);
\draw (2,-.5) rectangle (4.5,-2.5) node[pos=.5] {$P(y_1,y_2|x_1,x_2)$};
\draw (4.5,-.75) -- (5.5,-.75) node[pos=0,above right] {$Y_1$};
\draw (4.5,-2.25) -- (5.5,-2.25) node[pos=0,above right] {$Y_2$};
\draw (5.5,.5) rectangle (9,-3.5) node[pos=0.5] {$N_{G_{\SV}}(z|x_1,y_1;x_2,y_2)$};
\draw (9,-1.5) -- (10,-1.5) node[pos=0.5,above right] {$Z$};
\end{tikzpicture}
\caption{Entanglement-assisted coding strategy for the multiple access channel based on a linear system game.\\
	The MAC $N_{G_{\SV}}$ is defined in terms of the linear system game $G_{\SV}$ discussed in \Cref{sec:linear-system-games}.
	The correlation $P$ produced by the quantum strategy is detailed in \cite{SV18}.}
\label{fig:ea-coding-strategy}
\end{figure}

In this section we prove the existence of a sequence of coding strategies for the MAC $N_{G_{\SV}}$ defined in terms of the $m\times n$-linear system game $G_{\SV}$ described above that achieves the rate pair $(\log m, \log n)$ in the achievable rate region in the limit $d\to\infty$.

\begin{prop}\label{prop:ideal-rate-achievable-limit}
	Let $G_{\SV}$ be the linear system game from \cite{SV18} associated with the $m\times n$-linear system $Ax=b$, and let $N_{G_{\SV}}$ be the MAC defined in terms of $G_{\SV}$ via \eqref{eq:non-local-mac}.
	Assume that the two players share a maximally entangled state $|\psi\rangle\in\mathbb{C}^d\ox \mathbb{C}^d$ of Schmidt rank $d$ sufficiently large.
	Then there is a coding strategy that achieves the rate pair $R = (R_1,R_2)$, where
	\begin{align}
	R_1 &= (1-p_L) \log m - (1-p_L) (f(d) \log(nm-1) + h(f(d)))  -\frac{p_L}{2} \log(nm-1) - h(p_L)\\
	R_2 &= (1-p_L) \log n - (1-p_L) (f(d) \log(nm-1) + h(f(d)))  -\frac{p_L}{2} \log(nm-1) - h(p_L).
	\end{align}
	For this coding strategy, both the losing probability $p_L$ and the function $f(d)$ vanish in the limit $d\to\infty$.
\end{prop}

\begin{proof}
	In order to prove the claim of the proposition, we make use of the following easily-verifiable entropic inequalities:
	\begin{align}
	I(A;B) - I(A;B|D) &= -H(A|B) + I(A;D) + H(A|BD) \geq -H(A|B) \label{eq:I}\\
	I(A;B|C) - I(A;B|CD) &= I(A;D|C) + H(D|ABC) - H(D|BC) \geq - H(D|BC) \label{eq:II}\\
	I(A;B|CD) - I(A;BC|D) &= -I(A;C|D)\label{eq:III}
	\end{align}
	
	Let $X_i,Y_i$ be the random variables associated to the questions and answers for players $i=1,2$, let $Z$ be the random variable associated to the output of the MAC $N_{G_{\SV}}$, and let $W$ be the random variable indicating a win defined in \eqref{eq:W}.
	We fix the following coding strategy: Alice and Bob draw the questions $x_1$ and $x_2$ uniformly at random, and produce $y_1$ and $y_2$ using the quantum strategy detailed in \cite{SV18} based on measuring a maximally entangled state $\psi$, as depicted in \Cref{fig:ea-coding-strategy}. 
	In terms of the general entanglement-assisted coding scenario described in \Cref{sec:EA-MAC}, this corresponds to setting $A_i=X_i$, $X'_i=Y_i$, and using the trivial post-processing $f_i(x_j,y_j|x_i,y_i) = \delta_{x_i,x_j}\delta_{y_i,y_j}$.
	By the right-hand inequality in \eqref{eq:bounds-app} (which is proved in Theorem 1.1 in \cite{SV18}), the above strategy has losing probability 
	\begin{align}
	p_L \leq \left(\frac{C'}{d}\right)^2 \label{eq:losing-prob-bound}
	\end{align} 
	for some constant $C'$.
	
	We first determine an achievable rate $R_1$ for the first sender satisfying $R_1\leq I(Z;X_1|X_2)$ (see \eqref{eq:cap-region-constraints} and the discussion in \Cref{sec:EA-MAC}).
	To this end, we use \eqref{eq:II} with the choices $A=Z$, $B=X_1$, $C=X_2$, $D=W$ to obtain
	\begin{align}
	I(Z;X_1| X_2) &\geq I(Z;X_1 | X_2 W) - H(W|X_1 X_2)\\
	&= I(Z; X_1 X_2|W) - I(Z;X_2|W) - H(W|X_1X_2)\\
	&\geq I(Z; X_1 X_2|W) - I(Z;X_2|W) - h(p_L)\\
	&= (1-p_L) \left[H(X_1 X_2|W=1) - H(X_2|W=1)\right] - h(p_L)\\
	&\geq (1-p_L) \left[ H(X_1 X_2|W=1) - \log n \right] -\frac{p_L}{2} \log(nm-1) - h(p_L).\label{eq:R1-bound}
	\end{align}
	In the second line we used \eqref{eq:III} and in the third line we used $H(W|X_1X_2)\leq H(W)\leq h(p_L)$.
	In the fourth line we used that, if Alice and Bob win the game ($W=1$), then the variable $Z$ is a deterministic function of $X_1X_2$ and hence $I(Z; X_1 Y_1 X_2 Y_2|W=1) = H(X_1 X_2|W=1)$, together with the fact that $I(Z; X_1 Y_1 X_2 Y_2|W=0) = 0$.
	Finally, in the last line we used the trivial bound $H(X_2|W=1) \leq \log |\cX_2| = n$ as well as the fact that $\frac{p_L}{2} \log(nm-1)\geq 0$.
	
	We now bound the entropy $H(X_1 X_2|W=1)$ in \eqref{eq:R1-bound} by considering the probability distribution 
	\begin{align}
	\pi_{X_1X_2}^W = \lbrace \Pr(X_1X_2=x_1 x_2|W=1)\rbrace_{x_1,x_2}.
	\end{align}
	Our goal is to show that $\pi_{X_1X_2}^W$ converges to the uniform distribution $\pi_U$ on $\cX_1\times\cX_2$ in total variation distance as $d\to\infty$.
	By continuity of entropy, this then implies that $H(X_1 X_2|W=1) \approx \log m + \log n$ with the approximation error vanishing in the limit $d\to\infty$.
	
	To show this claim, we use Bayes' theorem to express $\Pr(X_1X_2=x_1 x_2|W=1)$ as
	\begin{align}
	\Pr(X_1X_2=x_1 x_2|W=1) &= \frac{\Pr(W=1|X_1X_2=x_1x_2) \Pr(X_1X_2=x_1x_2)}{\Pr(W=1)}\\
	&= \frac{1}{nm}\frac{\Pr(W=1|X_1X_2=x_1x_2)}{\Pr(W=1)}.
	\end{align}
	Due to \eqref{eq:losing-prob-bound}, the winning probability satisfies 
	\begin{align}
	\Pr(W=1)=1-p_L \geq 1-(C'/d)^2.
	\label{eq:winning-probability-bound}
	\end{align}
	Moreover, by Lemma 4.2 in \cite{SV18} every strategy that achieves a winning probability of at least $1-p_L$ wins with probability $1-nm p_L$ on any question $(x_1,x_2)$, and hence 
	\begin{align}
	\Pr(W=1|X_1X_2=x_1x_2) \geq 1-nm p_L \geq 1-nm(C'/d)^2.\label{eq:winning-probability-questions-bound}
	\end{align}
	For the total variation distance $\dtv(\pi_{X_1X_2}^W,\pi_U)$, the bounds \eqref{eq:winning-probability-bound} and \eqref{eq:winning-probability-questions-bound} imply that
	\begin{align}
	\dtv(\pi_{X_1X_2}^W,\pi_U) &= \frac{1}{2} \sum_{x_1,x_2} \left| \Pr(X_1X_2=x_1x_2|W=1) - \frac{1}{nm} \right| \eqqcolon f(d)
	\end{align}
	for some non-negative function $f(d)$ that converges to zero as $d\to \infty$.
	By the continuity of entropy \cite{Zhang2007},
	\begin{align}
	H(X_1 X_2|W=1) \geq \log m + \log n - f(d) \log(nm-1) - h(f(d)),\label{eq:entropy-continuity}
	\end{align}
	and substituting this in \eqref{eq:R1-bound} yields $R_1 \leq I(Z;X_1| X_2)$ with
	\begin{align}
	R_1 \coloneqq (1-p_L) \log m - (1-p_L) (f(d) \log(nm-1) + h(f(d)))  -\frac{p_L}{2} \log(nm-1) - h(p_L).\label{eq:R1}
	\end{align}
	Using similar steps as above, we can also show that $R_2 \leq I(Z;X_2| X_1)$ with
	\begin{align}
	R_2 \coloneqq (1-p_L) \log n - (1-p_L) (f(d) \log(nm-1) + h(f(d))) -\frac{p_L}{2} \log(nm-1) - h(p_L).\label{eq:R2}
	\end{align}
	
	For the rate pair $(R_1,R_2)$ to be achievable, it remains to be shown that $R_1+R_2$ satisfies the sum rate constraint $R_1+R_2\leq I(Z;X_1X_2)$.
	To this end, we use \eqref{eq:I} with the choices $A=Z$, $B=X_1X_2$, $D=W$ to obtain
	\begin{align}
	I(Z; X_1 X_2) &\geq I(Z; X_1 X_2 | W) - H(Z| X_1 X_2)\\
	&= (1-p_L) H(X_1X_2|W=1) - H(Z| X_1 X_2)\\
	&\geq (1-p_L) (\log m + \log n) - (1-p_L)(f(d) \log(nm-1) + h(f(d))) - H(Z| X_1 X_2),\label{eq:sum-rate-bound}
	\end{align}
	which follows from the discussion above and \eqref{eq:entropy-continuity}.
	To bound the conditional entropy $H(Z| X_1 X_2)$, note that $\Pr(Z\neq X_1 X_2) = p_L \frac{nm-1}{nm}\leq p_L$, and hence we can apply Fano's inequality to obtain the bound
	\begin{align}
	H(Z|X_1X_2) \leq p_L \log(nm-1) + h(p_L).
	\end{align}
	Substituting this in \eqref{eq:sum-rate-bound} yields
	\begin{align}
	I(Z; X_1 X_2) &\geq (1-p_L) (\log m + \log n) - (1-p_L)(f(d) \log(nm-1) + h(f(d))) - p_L \log(nm-1) - h(p_L)\\
	&\geq (1-p_L) (\log m + \log n) - 2(1-p_L) (f(d) \log(nm-1) + h(f(d))) - p_L \log (nm-1) - 2h(p_L)\\
	&= R_1 + R_2
	\end{align}
	with $R_1$ and $R_2$ as in \eqref{eq:R1} and \eqref{eq:R2}, respectively.
	This finishes the proof.
\end{proof}

By \Cref{prop:ideal-rate-achievable-limit}, the achievable rate region $\CeaO[d](N_{G_{\SV}})$ gets arbitrarily close to the rate pair $(\log m, \log n)$ in the limit $d\to\infty$.
Hence, we have the following result:
\begin{cor}
	Let $G_{\SV}$ be the linear system game from \cite{SV18} associated with the $m\times n$-linear system $Ax=b$, and let $N_{G_{\SV}}$ be the MAC defined in terms of $G_{\SV}$ via \eqref{eq:non-local-mac}.
	Then the rate pair $(\log m,\log n)$ is contained in the closure of $\CeaO (N_{G_{\SV}})$.
\end{cor}

\subsection{Undecidability of the Rate Region of a MAC}\label{sec:undecidability}

\Cref{prop:limited-entanglement-assistance,prop:ideal-rate-achievable-limit} show that there is a MAC $N_{G_{\SV}}$ defined in terms of the $m\times n$-linear system game $G_{\SV}$ such that the sum rate capacity $\CsumEA[d](N_{G_{\SV}})$ is bounded away from $(\log m,\log n)$ for any finite $d$, but the boundary of the $d$-entanglement-assisted single-letter capacity region $\CeaO[d](N_{G_{\SV}})$ gets arbitrarily close to $(\log m,\log n)$ in the limit $d\to\infty$.

Using a recent result by \textcite{Slofstra2019}, we can even prove the following: for a general linear system game $G_{\LS}$ and the corresponding MAC $N_{G_{\LS}}$, it is undecidable to determine if the maximal rate $(\log m,\log n)$ is achievable with finite-dimensional entanglement assistance:
\begin{prop}\label{prop:undecidability}
	Let $N_{G_{\LS}}$ be the MAC defined via \eqref{eq:non-local-mac} in terms of an $m\times n$-linear system game $G_{\LS}$.
	Then it is undecidable to determine if the rate pair $(\log m,\log n)$ belongs to $\CeaO(N_{G_{\LS}})$.
\end{prop}
\begin{proof}
	Let $G_{\LS}$ be the game associated to the $m\times n$-linear system $Ax=b$.
	Then Corollary 1.3 in \cite{Slofstra2019} proves that it is undecidable to determine if $G_{\LS}$ has a perfect strategy in the set of finite-dimensional quantum correlations as defined in \eqref{eq:entangled-correlation}.
	If there is a perfect strategy, then by the construction of $N_{G_{\LS}}$ the two senders can code at the rate pair $(\log m, \log n)$ by drawing the questions $x_i$ uniformly at random and using the perfect strategy to produce $y_i$ such that $(x_1,y_1,x_2,y_2)\in W$.
	Conversely, if there is no perfect strategy and hence $\omega_U(G_{\LS})<1$, then for any finite $d$ the sum rate capacity $\CsumEA[d](N_{G_{\LS}})$ can be bounded away from $\log m + \log n$ using \Cref{prop:classical-bound}, and this separates the point $(\log m,\log n)$ from the entanglement-assisted rate-region $\CeaO(N_{G_{\LS}})$.
	Hence, the pair $(\log m,\log n)$ belongs to $\CeaO(N_{G_{\LS}})$ if and only if there is a perfect strategy for $G_{\LS}$, which is undecidable.
\end{proof}

\section{Hardness of Computing the Capacity Region of MACs}\label{sec:hardness}
Despite the availability of a single-letter characterization, as given by \eqref{eq:cap-region-constraints}, computing the capacity region of an arbitrary multiple access channel is a difficult task \cite{BW11}. The difficulty lies in the inherent non-convexity of the problem, i.e., the optimization is constrained to be over product distributions \cite{CPFV10}.  In this section, we show that deciding if a MAC can be used perfectly or not (up to $\Theta (\frac{1}{n^3})$) is NP-hard. This implies that deciding if an arbitrary point $(R_1, R_2)$ belongs to the capacity region to within an additive error of $\Theta (\frac{1}{n^3})$ is NP-hard.

\subsection{The PCP Theorem}
The results to follow rely on the probabilistically checkable proofs (PCP) theorem, which says that any language in the class NP admits a characterization via probabilistically checkable proofs \cite{ALMSS98, D07, V01}. More formally, let PCP$_{c, s}[
r(n), q(n)]$ be the class of all languages $L$ such that there exists a verifier $V$, which is free to use $\mathcal{O}(r(n))$ random bits and query a given proof $\mathcal{O}(q(n))$ times, with the following properties:
\begin{enumerate}
\item Completeness: If $x \in L$, then there exists a proof $P$ such that $V$ accepts with probability at least $c$.
\item Soundness: If $x \notin  L$, then $V$ accepts with probability at most $s$.
\end{enumerate}
Note that this can be considered a generalization of NP as NP = PCP$_{1,0} [0, \text{poly}(n)]$. The original PCP theorem says that $\text{NP} \subseteq \text{PCP}_{1,1/2}[\log{n}, 1]$ \cite{Arora1992}. To illustrate its implications, consider the canonical NP-complete language 3SAT for example. Take a Boolean formula $\psi$ in 3-conjunctive normal form (3CNF), i.e., it is a conjunction of clauses that are disjunctions of three literals. Note that a literal can be a Boolean variable or its negation. Say $\psi \notin $ 3SAT. A verifier exists such that, with access to logarithmic randomness and a constant number of queries to a given proof or witness, it will reject with non-trivially high probability. This suggests that proving a falsehood, e.g., $\psi \in$ 3SAT when that is not the case, typically involves making many errors. 

The PCP theorem can be equivalently formulated as a statement about the hardness of approximating NP-complete problems \cite{Feige1991, H01}. We will restrict our attention to the following formulation. 

\begin{thm}[PCP theorem; \cite{F98, SS06}]
	Given a 3-CNF-5 Boolean formula $\psi$, to decide whether $\psi$ has a satisfying assignment or that every assignment violates at least $(1 - c)$ fraction of the clauses in $\psi$ is NP-hard, for some constant $c < 1$.
	\label{thm:pcp-theorem}
\end{thm}

\noindent Here, a formula $\psi$ is called 3-CNF-5 if it is a conjunction of $m$ clauses and each clause is a disjunction of exactly three distinct literals and each of the $n$ Boolean variables appears in exactly five clauses. Remark that the number of clauses $m$ is $\mathcal{O}(n)$. We call $\psi$ at most $c$-satisfiable for some $c \in [0,1]$ if some assignment satisfies $f$ fraction of its clauses, for  $f \in [0,c]$, and no assignments satisfies more than $c$ fraction of its clauses.

\subsection{The Basic Two-Prover Game}
We denote by $G_{\Ha}$ the non-local game version of the basic two-prover protocol introduced in \cite{H01}. Namely, given a 3-CNF-5 Boolean formula $\psi = C_{1} \wedge C_{2} \wedge ... \wedge C_{m}$ as input, where $C_{j} = y_{a_j} \lor y_{b_j} \lor y_{c_j}$, the referee does the following:
\begin{enumerate}
    \item Choose an integer $j \in \{1, $...$ , m\}$ uniformly at random and send $j$ to Alice. Choose $k \in \{a_j, b_j, c_j\}$ uniformly at random and send $k$ to Bob. 
    \item Receive an assignment for $C_j$ from Alice and a truth value for $x_k$ from Bob. They win if Alice's answer satisfies $C_j$ and the two agree on the value of $x_k$, otherwise they lose. 
\end{enumerate}
Let $\psi$ be at most $c$-satisfiable. Because the optimal strategy is deterministic, Bob will have an assignment to $\psi$. If the clause in the question to Alice is violated by Bob's assignment, then the best Alice can do is disagree with Bob on the value of one Boolean variable in the clause and hope that Bob did not receive it as a question. This implies that ${\omega}(G_{\Ha}) \leq \frac{2 + c}{3}$. Conversely, ${\omega}(G_{\Ha}) \leq \frac{2 + c}{3}$ implies that $\psi$ is at most $c$-satisfiable. To see this, note that if some assignment satisfies more than $c$ fraction of the clauses in $\psi$, then Alice and Bob can use it to win with probability higher than $\frac{2 + c}{3}$. Using the PCP theorem, these observations, in addition to the fact that 
$\psi \in 3SAT 	\Leftrightarrow {\omega}(G_{\Ha}) = 1$, imply that it is NP-hard to decide if $G_{\Ha}$ can be won with probability one or with probability at most $\frac{2 + c}{3}$.

\subsection{Hardness Result}
If the game is made promise-free, then it follows that it is NP-hard to decide if $\omega_{U}(G_{\Ha}) = 1$ or $\omega_{U}(G_{\Ha}) \leq 1 - (\frac{1-c}{n})$.

\begin{prop}\label{prop: Hardness-result}
It is NP-hard to decide if the sum capacity of the MAC associated with the promise-free version of $G_{\Ha}$ is equal to its maximum value $\log{m} + \log{n}$ or it is bounded away from it by $\Theta (\frac{1}{n^3})$.
\end{prop}

\begin{proof}[Proof of \Cref{prop: Hardness-result}]
Observe that if $\psi$ has a satisfying assignment, then the two senders can use the channel perfectly, i.e., $R_1 = \log m$ and $R_2 = \log{n}$. On the other hand, if $\psi$ has no satisfying assignment, then $\omega_{U}(G_{\Ha})$ is strictly less than 1. Hence, we can use \cref{prop:classical-bound} to make statements about the sum capacity in a manner similar to \eqref{prop:limited-entanglement-assistance}. Let $\delta = \frac{(1 - c)}{n^3}$ and assume that $\eps^* < \delta$. For large $n$, we can overestimate the left-hand side of \cref{eq:eps-from-delta} by 
\begin{align}
\frac{\delta+b \delta^{\beta}}{1-\delta} \geq \delta(\eps^*\|1-\omega_{U}(G_{\Ha}))	,
\end{align}
where $\beta = \frac{5}{6}$ and $b$ is taken to be large enough. Again, we use Pinsker's inequality to lower bound the right hand side.
\begin{align}
\frac{\delta+b \delta^{\beta}}{1-\delta} \geq 
\frac{2(1-c)^{2}}{\ln{2}} \Big[ \frac{1}{n^2} - 2 \frac{1}{n^4} \Big]
\end{align}
As $n$ goes to infinity, the right-hand side goes as $1/n^2$, while the left-hand side goes as $1/n^{2.5}$. $\eps^*$ cannot be smaller than $\delta$ for large enough $n$. Therefore, whenever $\omega_{U}(G_{\Ha}) < 1$, i.e., $\psi$ has no satisfying assignment, we conclude from \cref{eq:capacity-upper-bound} that for all large enough $n$, 

\begin{align}
I(X_1 Y_1 X_2 Y_2;Z) \leq \log m + \log n - 
\frac{(1 - c)}{n^3}
\end{align}
The proposition follows from here via the PCP theorem. 
\end{proof}
It is instructive to compare this hardness result with the time complexity of the popular Arimoto-Blahut (AB) algorithm for computing the point-to-point discrete channel capacity \cite{Arimoto72, Blahut72}. If we consider the two senders together, then the channel capacity is the solution to a convex program. The number of iterations needed in order to have $\frac{\epsilon}{n^3}$ additive precision for the capacity using the AB algorithm is $\frac{\mathcal{O}(n^3 \log{n})}{\epsilon}$ in the worst case. Assuming P $\neq$ NP, there is no polynomial-time algorithm to get to within the same precision for the boundary of the capacity region of an arbitrary discrete MAC. Moreover, assuming the exponential time hypothesis \cite{ETH}, there is no sub-exponential algorithm to compute the boundary of the region to inverse cubic precision. In such a case, one may consider the ``naive'' method of covering the space of product probability distributions with a net and computing an approximation of the capacity region.
We argue below that, assuming the validity of the exponential time hypothesis, this net covering method is not far from optimal.

Let $K\subseteq \mathbb{R}^n$ be a subset of the Euclidean space $\mathbb{R}^n$ and $\eps>0$.
An \textit{$\eps$-net} for $K$ is a subset $N\subseteq K$ such that every point of $K$ is within distance $\eps$ of a net point in $N$.
We denote by $C(K,\eps)$ the \textit{covering number} of $K$, defined as the smallest possible cardinality of an $\eps$-net $N$ for $K$.
By a standard volume argument, $C(K,\eps)$ is bounded from below as
\begin{align}
C(K,\eps) \geq \frac{|K|}{|\cB_\eps^n|},\label{eq:covering}
\end{align}
where $|K|$ denotes the (Euclidean) volume of $K$ embedded in $\mathbb{R}^n$, and $\cB_\eps^n$ is the $n$-ball with radius $\eps$.
Let now $K=\Delta_n$ be the $n$-probability simplex, and recall that 
\begin{align}
|\Delta_n| = \frac{\sqrt{n}}{(n-1)!} \qquad \text{and}\qquad |\cB_\eps^n| = \frac{\pi^{n/2}}{\Gamma(\frac{n}{2}+1)} \eps^n.
\end{align}
Here, $\Gamma(\cdot)$ is the well-known Gamma function, satisfying $\Gamma(n)=(n-1)!$ for $n\in\mathbb{N}$.
Using the Stirling approximation $n!\sim \sqrt{2\pi n} \left(\frac{n}{e}\right)^n$ as well as $\eps = \poly(n)^{-1}$, we obtain from \eqref{eq:covering} that $C(K,\eps) =\Omega(\poly(n)^n)$.

\end{document}